\documentclass[final,11pt]{article}
\pdfoutput=1
\usepackage[utf8]{inputenc}

\usepackage[thmstyles]{andrews}
\usepackage{andrews_bib}
\usepackage[inline]{showlabels}

\usepackage[style=alphabetic, 
backref, 
url=false,
doi=false,
isbn=false,
useprefix=true,
maxcitenames=99,
maxbibnames=99,
maxalphanames=99,
maxnames=99]{biblatex}
\newbibmacro{string+doi}[1]{%
  \iffieldundef{doi}{#1}{\href{http://dx.doi.org/\thefield{doi}}{#1}}}
\DeclareFieldFormat{title}{\usebibmacro{string+doi}{\mkbibemph{#1}}}
\DeclareFieldFormat*{title}{\usebibmacro{string+doi}{\mkbibquote{#1}}}

\newcommand{\pathrm}{\mathrm{path}}

\title{Algebraic Hardness versus Randomness in Low Characteristic}
\author{Robert Andrews\thanks{Department of Computer Science, University of Illinois at Urbana-Champaign. Email: \texttt{rgandre2@illinois.edu}. Supported by NSF grant CCF-1755921.}}
\date{May 21, 2020}

\begin{document}

\maketitle

\begin{abstract}
	We show that lower bounds for explicit constant-variate polynomials over fields of characteristic $p > 0$ are sufficient to derandomize polynomial identity testing over fields of characteristic $p$.
	In this setting, existing work on hardness-randomness tradeoffs for polynomial identity testing requires either the characteristic to be sufficiently large or the notion of hardness to be stronger than the standard syntactic notion of hardness used in algebraic complexity.
	Our results make no restriction on the characteristic of the field and use standard notions of hardness.

	We do this by combining the Kabanets-Impagliazzo generator with a white-box procedure to take $p$\ts{th} roots of circuits computing a $p$\ts{th} power over fields of characteristic $p$.
	When the number of variables appearing in the circuit is bounded by some constant, this procedure turns out to be efficient, which allows us to bypass difficulties related to factoring circuits in characteristic $p$.

	We also combine the Kabanets-Impagliazzo generator with recent ``bootstrapping'' results in polynomial identity testing to show that a sufficiently-hard family of explicit constant-variate polynomials yields a near-complete derandomization of polynomial identity testing.
	This result holds over fields of both zero and positive characteristic and complements a recent work of Guo, Kumar, Saptharishi, and Solomon, who obtained a slightly stronger statement over fields of characteristic zero.
\end{abstract}

\section{Introduction}

The interaction between computational hardness and pseudorandomness is a central theme of computational complexity.
The goal of this vein of work is to show that a class $\mathcal{C}$ of problems that are solvable by randomized algorithms can in fact be solved by deterministic algorithms which are not much slower than the known randomized algorithm, assuming lower bounds for a related class $\mathcal{D}$.
When trying to derandomize $\BPP$, the class of problems solvable in polynomial time by a randomized Turing machine with failure probability at most $1/3$, we understand this problem quite well.
A series of works culminated in that of \textcite{IW97}, which showed that $\BPP = \P$ if there are problems in $\mathsf{E}$ which require boolean circuits of exponential size.
Subsequent work by \textcite{SU05, Uma03} further tightened the quantitative tradeoffs obtainable for derandomizing $\BPP$.

In this work, we focus on the question of hardness versus randomness in the more restricted computational model of algebraic circuits, which naturally compute multivariate polynomials over a specified base field $\F$.
Here, the algorithmic problem of interest is \emph{polynomial identity testing} (PIT), which is the problem of determining if a given algebraic circuit computes the identically zero polynomial.
We typically consider identity testing of circuits whose size and degree are bounded by a polynomial function in the number of variables.
This low-degree regime captures polynomials of interest to computer scientists, such as the determinant and permanent, and corresponds to typical algorithmic applications of PIT.
In this regime, the problem of PIT is easily solved with randomness by evaluating the circuit at a randomly chosen point of a large enough grid.
The correctness of this algorithm follows from the Schwartz-Zippel lemma, which roughly says that a low-degree multivariate polynomial cannot vanish at many points of a sufficiently large grid.
To date, no deterministic algorithm for PIT is known that substantially improves on the na\"ive derandomization of the Schwartz-Zippel lemma.

Polynomial identity testing has widespread applications in theoretical computer science and has led to randomized algorithms for perfect matching \cite{Lov79, KUW86, MVV87}, primality testing \cite{AB03,AKS04}, and equivalence testing of read-once branching programs \cite{BCW80}, among other problems.
In light of the utility of PIT as an algorithmic primitive, it is worth understanding to what extent PIT can be derandomized.
There is a large body of work concerned with unconditional derandomization of PIT for various sub-classes of algebraic circuits.
For more on this, we refer the reader to the surveys of \textcite{SY10,Saxena09,Saxena14}.
In this work, we will focus on conditional derandomization of PIT under suitable hardness assumptions.

\subsection{Prior Work}

The first instantiation of the hardness-randomness paradigm for polynomial identity testing was given by \textcite{KI04}.
Their work implemented the design-based approach of \textcite{NW94} in the algebraic setting, showing that lower bounds for an explicit family of multivariate polynomials can be used to derandomize PIT.

Subsequent work by \textcite{DSY09,CKS18} extended this to the setting of bounded-depth circuits, roughly showing that lower bounds against depth-$(\Delta + O(1))$ circuits suffice to derandomize identity testing of depth-$\Delta$ circuits, for any constant $\Delta$.
The result of \textcite{DSY09} works with any hard polynomial, but scales poorly with the individual degree of the circuit being tested.
\textcite{CKS18} refined the approach of \textcite{DSY09} and showed that if the family of hard polynomials has sufficiently low degree, then this dependence on the individual degree of the circuit being tested can be avoided.
Implementing the hardness-randomness paradigm in low-depth is motivated in part by a host of depth-reduction results in algebraic complexity \cite{AV08, Koi12, Tav15, GKKS16} which show that polynomials computable by small circuits can be computed by non-trivially small low-depth circuits.

Returning to the setting of unrestricted circuits, recent work of \textcite{GKSS19} uses a stronger hardness assumption than that of \textcite{KI04} and obtains a stronger derandomization of PIT.
Specifically, \textcite{GKSS19} obtain a polynomial-time derandomization of PIT using lower bounds against an explicit family of constant-variate polynomials.
For comparison, \textcite{KI04} only obtain quasipolynomial-time algorithms for PIT under multivariate hardness assumptions.
In \autoref{sec:kronecker} of this work, we further discuss the relationship between these hardness assumptions and provide evidence for the strength of constant-variate hardness compared to multivariate hardness.

A separate line of work by \textcite{AGS19,KST19} shows that PIT exhibits a ``bootstrapping'' phenomenon.
That is, if one can obtain a barely non-trivial derandomization of PIT for circuits of size and degree which are unbounded in the number of variables, then it follows that there is a near-complete derandomization of PIT for circuits of polynomial size and degree.

From these works, we have a relatively good understanding of what derandomization of PIT is possible under hardness assumptions.
However, excluding the bootstrapping results of \textcite{AGS19,KST19}, all previous work on hardness-randomness tradeoffs for PIT requires the underlying field to be of zero or large characteristic (for the definition of the characteristic of a field, see \autoref{sec:prelim}).
That is, we can derandomize PIT under hardness assumptions over the complex numbers $\mathbb{C}$ or the finite field of $p^m$ elements $\F_{p^m}$ when $p$ is sufficiently large, but we do not know how to do the same over a field of low characteristic like $\F_{2^m}$.

A partial exception to this deficiency is the work of \textcite{KI04}.
Their results yield derandomization of PIT over a finite field $\F_{p^m}$ assuming an explicit polynomial which is hard to compute \emph{as a function} over $\F_{p^m}$.
Over infinite fields, two polynomials are equal if and only if they compute the same function.
However, this no longer holds over finite fields.
For example, over $\F_2$, the polynomial $x^2 - x$ computes the zero function but is decidedly not the zero polynomial.
It is more common in the study of algebraic circuits to prove lower bounds on the task of computing a polynomial as a syntactic object, not as a function.
Functional lower bounds imply syntactic lower bounds, but the reverse direction does not hold, which makes proving functional lower bounds a harder task.

If one inspects the proof of \textcite{KI04}, the functional hardness assumption can be replaced with a slightly weaker, albeit non-standard, syntactic hardness assumption.
Namely, it suffices to assume the existence of an explicit family of $n$-variate polynomials $\set{f_n : n \in \naturals}$ such that $f_n^{p^k}$ is hard in the syntactic sense for $1 \le p^k \le 2^{O(n)}$.
Over characteristic zero fields, the factoring algorithm of \textcite{Kaltofen89} implies that if $f$ is hard to compute, then $f^d$ is comparably hard to compute as long as $d$ is not too large.
Over fields of characteristic $p$, it is not clear if hardness of $f^p$ is implied by hardness of $f$.
For example, it is consistent with our current state of knowledge that the $n \times n$ permanent $\perm_n(\vec{x})$ is $2^{\Omega(n)}$-hard over $\F_3$, but that $\perm_n(\vec{x})^3$ is computable by circuits of size $O(n^2)$ over $\F_3$.
Understanding the relationship between the complexity of $f$ and $f^p$ over fields of characteristic $p > 0$ in general remains a challenging open problem.

For further exposition on hardness-randomness tradeoffs for PIT, see the recent survey of \textcite{KS19}.

\subsection{Identity Testing in Low Characteristic}

Before describing our contributions, we take a detour to look more closely at the question of derandomizing PIT over fields of low characteristic.
Known techniques for derandomizing PIT over fields of small characteristic under hardness assumptions fail due to the fact that over a field of positive characteristic, the derivative of a non-constant polynomial may be zero.
For example, over $\F_2$, we have $\pd{}{x}(x^2) = 2x = 0$, since $2=0$ in $\F_2$.
Thus, techniques which are in some sense ``analytic'' break in low characteristic.
Given that the problem of polynomial identity testing is entirely algebraic, it would be nice to find an ``algebraic'' approach that does not succumb to this flaw.
Indeed, derandomizing PIT in low characteristic fields under hardness assumptions is listed as an open problem in the recent survey of \textcite{KS19} on algebraic derandomization.

The problem of derandomizing PIT in low characteristic fields also has interesting algorithmic applications.
Consider, for example, the randomized algorithm of \textcite{Lov79} to detect whether a bipartite graph has a perfect matching.
Let $G = (V_1 \sqcup V_2, E)$ be a balanced bipartite graph on $2n$ vertices with partite sets $V_1$ and $V_2$.
We form the $n \times n$ symbolic matrix $A$ given by
\[
	A_{i,j} = \begin{cases}
		x_{i,j} & \set{i,j} \in E \\
		0 & \text{otherwise}.
	\end{cases}
\]
It is not hard to see that $\det(A) \neq 0$ if and only if $G$ has a perfect matching.
We can then check if $G$ has a perfect matching by evaluating $A$ at a random point chosen from a suitably large grid of integers.

In evaluating $\det(A)$, we may encounter large numbers of size $\Omega(n!)$.
Arithmetic on such numbers is expensive, requiring at least $\Omega(n \log n)$ time.
We could instead implement this algorithm over a finite field of size $\poly(n)$.
As the determinant is a polynomial of degree $n$, the Schwartz-Zippel lemma guarantees that this modification yields an algorithm with low error probability.
What we have gained is the fact that elements of such a finite field can be represented in $O(\log n)$ bits, so our arithmetic becomes more efficient.
In principle, one could choose the field so that the characteristic is large enough for the the hardness-randomness paradigm to apply, but there may be other considerations which motivate picking, say, an extension field of $\F_2$.
Derandomizing such an algorithm (under hardness assumptions) requires extending the hardness-randomness paradigm to fields of low characteristic.

Alternatively, one can reduce the bit complexity by using a derandomized polynomial identity testing algorithm over the rational numbers, but with the arithmetic performed modulo a small prime number.
This approach also achieves logarithmic bit complexity.
However, we are now in the position of having to derandomize the selection of the prime number.
It is not obvious how to do this much faster than brute force, so the benefits of reducing the bit complexity are negated by the need to try many different primes.

While the previous example may seem somewhat artificial, we remark that there are instances of algorithms which explicitly rely on polynomial identity testing over fields of low characteristic.
For example, the randomized algorithm of \textcite{Will09} for the $k$-path problem makes use of polynomial identity testing over fields of characteristic 2.
If one wanted to derandomize this algorithm under a hardness assumption, prior work on hardness-randomness tradeoffs for PIT would not suffice.

\subsection{Our Results}

In this work, we instantiate the hardness-randomness paradigm for PIT over fields of low characteristic under standard syntactic hardness assumptions.
That is, we obtain derandomization of PIT from the existence of an explicit family of hard polynomials $\set{f_n : n \in \naturals}$ without assuming hardness of $p$\ts{th} powers of $f_n$.
At the heart of our results is a new technique for computing the map $f^p \mapsto f$ over $\F[\vec{x}]$ when the polynomial $f^p$ is given by an algebraic circuit.
When $f$ depends on a small number of variables, the circuit computing $f$ is not too much larger than the circuit which computes $f^p$.

\begin{lemma}[informal version of \autoref{cor:pth root of circuit}] \label{lem:informal pth root}
	Suppose $f(\vec{x})^p$ is a polynomial on $O(1)$ variables and can be computed by a circuit of size $s$ over a field of characteristic $p > 0$.
	Then $f(\vec{x})$ can be computed by a circuit of size $O(s)$.
\end{lemma}

Using this, we are able to extend the techniques of \textcite{KI04} to fields of low characteristic.
To do so, we need stronger hardness assumptions than those made by \textcite{KI04} for the case of zero characteristic fields.
In algebraic complexity, lower bounds are typically proved for families of polynomials parameterized by the number of variables, as this captures the regime of interest for algorithmic applications.
To prove our results, we assume lower bounds against a family of constant-variate polynomials which are parameterized by degree.

For the sake of exposition, we focus on the case of lower bounds for univariate polynomials.
A univariate polynomial of degree $d$ can easily be computed by circuits of size $O(d)$ using Horner's rule.
It is not hard to show that every such polynomial also requires size $\Omega(\log d)$ to compute.
However, improving on this $\Omega(\log d)$ lower bound for an explicit family of polynomials is a long-standing open problem.
Standard dimension arguments show that most univariate polynomials of degree $d$ require circuits of size $d^{\Omega(1)}$ to compute.

When comparing statements regarding degree $d$ univariates and degree $n^{O(1)}$ multivariate polynomials on $n$ variables, it is instructive to think of $n$ and $\log d$ as comparable.
In this sense, our results achieve the same hardness-randomness tradeoffs as those of \textcite{KI04}, but require translating their hardness assumptions to the comparable statement for univariate polynomials.

Using \autoref{lem:informal pth root}, we can extend the analysis of \citeauthor{KI04} to work over fields of low characteristic.
We now give two concrete examples of the derandomization we can obtain using this extension.

\begin{theorem}[informal version of \autoref{thm:KI hard vs rand} and \autoref{cor:KI tradeoffs}]
	Let $\F$ be a field of characteristic $p > 0$.
	Let $\set{f_d(x) : d \in \naturals}$ be an explicit family of univariate polynomials which cannot be computed by circuits of size less than $s(d)$ over $\F$.
	\begin{enumerate}
		\item
			If $s(d) = \log^{\omega(1)} d$, then there is a deterministic algorithm for identity testing of polynomial-size, polynomial-degree circuits over $\F$ in $n$ variables which runs in time $2^{n^{o(1)}}$.
		\item
			If $s(d) = 2^{\log^{\Omega(1)} d}$, then there is a deterministic algorithm for identity testing of polynomial-size, polynomial-degree circuits over $\F$ in $n$ variables which runs in time $2^{\log^{O(1)} n}$.
	\end{enumerate}
\end{theorem}

For comparison, from an $n^{\omega(1)}$ lower bound against a family of explicit multilinear polynomials, \textcite{KI04} give a deterministic algorithm for PIT over fields of characteristic zero which runs in time $2^{n^{o(1)}}$.
If instead one has a $2^{n^{\Omega(1)}}$ lower bound, then their techniques yield a deterministic algorithm which runs in time $2^{\log^{O(1)}n}$.
Viewing $\log d$ and $n$ as (roughly) equivalent, we see that our derandomization obtains the same tradeoff between hardness and pseudorandomness as \textcite{KI04}, modulo the difference between univariate and multivariate lower bounds.

It is not hard to show that lower bounds in the constant-variate regime imply comparable lower bounds in the multivariate regime (see \autoref{lem:kronecker}), but the reverse implication is not known.
In \autoref{sec:kronecker}, we investigate the possibility of using known techniques to prove univariate lower bounds from multivariate lower bounds.

As the assumption of a hard univariate family seems strong, it raises the question of whether or not one can obtain a stronger derandomization of PIT over fields of positive characteristic under a univariate hardness assumption.
There is evidence this can be done, as \textcite{GKSS19} use univariate lower bounds to obtain a complete derandomization of PIT over fields of characteristic zero.
With a more careful instantiation of the Kabanets-Impagliazzo result, we are able to derandomize PIT in a way that suffices for the bootstrapping results of \textcite{AGS19,KST19} to take effect.
This allows us to prove nearly-optimal hardness-randomness tradeoffs for PIT over fields of positive characteristic, which comes close to matching the characteristic zero result of \textcite{GKSS19}.
More concretely, we prove the following.

\begin{theorem}[informal version of \autoref{thm:bootstrap generator}]
	Let $\F$ be a field of characteristic $p > 0$.
	Let $\set{f_d(x) : d \in \naturals}$ be an explicit family of univariate polynomials which cannot be computed by circuits of size less than $d^{\delta}$ for some constant $\delta > 0$.
	Then there is a deterministic algorithm for identity testing of polynomial-size, polynomial-degree algebraic circuits in $n$ variables over $\F$ which runs in time $n^{\exp \circ \exp(O(\log^\star n))}$.
\end{theorem}

The rest of this work is organized as follows.
In \autoref{sec:prelim}, we establish notation, definitions, and relevant background necessary to state and prove our results.
In \autoref{sec:pth root}, we prove our main technical lemma on computing $p$\ts{th} roots of algebraic circuits over fields of characteristic $p > 0$.
We then use this in \autoref{sec:KI} to extend the work of \citeauthor{KI04} to the low characteristic setting.
We combine our techniques with the bootstrapping results to obtain near-complete derandomization of PIT over fields of positive characteristic in \autoref{sec:bootstrap}.
\autoref{sec:kronecker} investigates the relationship between univariate and multivariate circuit lower bounds.
We conclude in \autoref{sec:conclusion} with a collection of problems left open by this work.

\section{Preliminaries} \label{sec:prelim}

For $n \in \naturals$, we write $[n] \coloneqq \set{1,\ldots,n}$ and $\llb n \rrb \coloneqq \set{0,\ldots,n-1}$.
If $A$ is an $n \times m$ matrix, we write $A_{i,\bullet}$ and $A_{\bullet,j}$ for the $i$\ts{th} row and $j$\ts{th} column of $A$, respectively.
We abbreviate a vector of variables $(x_1,\ldots,x_n)$, numbers $(a_1,\ldots,a_n)$, or field elements $(\alpha_1,\ldots,\alpha_n)$ by $\vec{x}$, $\vec{a}$, and $\vec{\alpha}$, respectively, where the length is usually clear from context.
We also abbreviate the product $\prod_{i=1}^n x_i^{a_i} \eqqcolon \vec{x}^{\vec{a}}$.
Given a polynomial $f(\vec{x}) = \sum_{\vec{a}} \alpha_{\vec{a}} \vec{x}^{\vec{a}}$, we write $\deg(f)$ and $\ideg(f)$ for the \emph{total degree} and \emph{individual degree} of $f$, respectively.
The total degree of $f$ is given by $\deg(f) \coloneqq \max \set{\norm{\vec{a}}_1 : \alpha_{\vec{a}} \neq 0}$, while the individual degree of $f$ is given by $\ideg(f) \coloneqq \max \set{\norm{\vec{a}}_\infty : \alpha_{\vec{a}} \neq 0}$.

For a field $\F$, the \emph{characteristic} of $\F$, denoted $\ch \F$, is the smallest positive integer $p$ such that $p \cdot 1 = 0$ in $\F$.
In the case that there is no such $p$, we say that $\F$ has characteristic zero.
Alternatively, $\ch\F$ is the number $p$ such that the ring homomorphism $\integers \to \F$ induced by $1 \mapsto 1$ has kernel $p \integers$.
The set $\mathcal{C}_{\F}(s,n,d) \subseteq \F[\vec{x}]$ denotes the set of all $n$-variate degree $d$ polynomials which can be computed by an algebraic circuit of size at most $s$ over $\F$.

\subsection{Algebraic Computation and Polynomial Identity Testing}

We assume familiarity with the models of algebraic circuits, formulae, and branching programs.
When we refer to the \emph{size} of a circuit, formula, or branching program, we mean the number of nodes in the computational device.
An introduction to this area can be found in the survey of \textcite{SY10}.
Throughout this work, we analyze our algorithms under the assumption that arithmetic over the base field $\F$ can be performed in constant time.

We now collect basic definitions and results needed for the study of deterministic black-box algorithms for polynomial identity testing.
More in-depth exposition is available in the recent survey of \textcite{KS19}.

We start with the notion of a hitting set, the basic object used to construct deterministic black-box algorithms for polynomial identity testing.

\begin{definition}
	Let $\mathcal{C} \subseteq \F[\vec{x}]$ be a set of $n$-variate polynomials.
	We say that a set $\mathcal{H} \subseteq \F^n$ is a \emph{hitting set for $\mathcal{C}$} if for every non-zero $f(\vec{x}) \in \mathcal{C}$, there is a point $\vec{\alpha} \in \mathcal{H}$ such that $f(\vec{\alpha}) \neq 0$.
	If $\mathcal{H}$ can be computed in $t(n)$ time, then we say that $\mathcal{H}$ is \emph{$t(n)$-explicit}.
\end{definition}

We now introduce hitting set generators, the analogue of pseudorandom generators in the context of algebraic derandomization.

\begin{definition}
	Let $\mathcal{C} \subseteq \F[\vec{x}]$ be a set of $n$-variate polynomials.
	Let $\mathcal{G} : \F^m \to \F^n$ be a mapping given by
	\[
		\mathcal{G}(\vec{y}) = (\mathcal{G}_1(\vec{y}), \ldots, \mathcal{G}_n(\vec{y})),
	\]
	where $\mathcal{G}_i \in \F[\vec{y}]$.
	We say that $\mathcal{G}$ is a \emph{hitting set generator for $\mathcal{C}$} if for every non-zero $f(\vec{x}) \in \mathcal{C}$, we have $f(\mathcal{G}(\vec{y})) \neq 0$.
	The \emph{seed length} of $\mathcal{G}$ is $m$.
	The \emph{degree} of $\mathcal{G}$ is $\max_{i \in [n]} \deg(\mathcal{G}_i)$.
	We say $\mathcal{G}$ is \emph{$t(n)$-explicit} if, given $\vec{\alpha} \in \F^m$, we can compute $\mathcal{G}(\vec{\alpha})$ in $t(n)$ time.
\end{definition}

It is a well-known result that an explicit, low-degree hitting set generator for $\mathcal{C}$ with small seed length yields an explicit hitting set for $\mathcal{C}$ of small size.
The hitting set is constructed by evaluating the generator on a grid of large enough size.
Correctness follows from the Schwartz-Zippel lemma.

\begin{lemma} \label{lem:hsg to hitting set}
	Let $\mathcal{C}$ be a set of $n$-variate degree $d$ polynomials.
	Let $\mathcal{G} : \F^m \to \F^n$ be a $t(n)$-explicit hitting set generator for $\mathcal{C}$ of degree $D$.
	Then there is a $(dD+1)^m t(n)$-explicit hitting set $\mathcal{H}$ for $\mathcal{C}$ of size $(dD+1)^m$.
\end{lemma}

We also need a notion of explicitness for a family of polynomials.
In previous works on hardness-randomness tradeoffs for polynomial identity testing, a family of $n$-variate polynomials $\set{f_n \in \F[\vec{x}] : n \in \naturals}$ is considered explicit if $f_n$ is computable in $\exp(O(n))$ time.
However, we will need a slightly different notion of explicitness.
Instead of an exponential-time algorithm to compute $f_n$, we require an exponential-time algorithm to compute the coefficient of a given monomial in $f_n$.
This different notion of explicitness will be used to transition between the constant-variate and multivariate regimes later on in \autoref{sec:KI} and \autoref{sec:bootstrap}.

\begin{definition}
	Let $\set{f_{n,d}(\vec{x}) \in \F[\vec{x}] : n, d \in \naturals}$ be a family of $n$-variate degree $d$ polynomials.
	We say that this family is \emph{strongly $t(n,d)$-explicit} if there is an algorithm which on input $(n,d,\vec{a})$ outputs the coefficient of $\vec{x}^{\vec{a}}$ in $f_{n,d}(\vec{x})$ in $t(n,d)$ time.
\end{definition}

\begin{remark}
	The preceding definition is reminiscent of Valiant's criterion for membership in $\VNP$.
	Briefly, Valiant's criterion says that if the coefficient of $\vec{x}^{\vec{a}}$ can be computed in $\#\P/\poly$, then the polynomial $f(\vec{x})$ is in $\VNP$, an algebraic analogue of $\NP$.
	We refer the reader to \textcite[Chapters 1 and 2]{Bur00} for further exposition on $\VNP$ and Valiant's criterion.
\end{remark}

We will repeatedly build explicit families of hard multivariate polynomials out of explicit families of hard constant-variate polynomials.
By ``a family of hard multivariate polynomials,'' we mean a family of polynomials $\set{f_n(\vec{x}) \in \F[\vec{x}] : n \in \naturals}$, where $f_n$ is an $n$-variate polynomial of degree $n^{O(1)}$.
When we say ``a family of hard constant-variate polynomials,'' we mean a family $\set{f_d(\vec{x}) \in \F[\vec{x}] : d \in \naturals}$, where $f_d$ is a degree $d$ polynomial on $k = O(1)$ variables.
That is, when we consider multivariate polynomials, we parameterize the family by the number of variables and primarily consider families of small degree; when we look at constant-variate polynomials, we fix the number of variables in all polynomials and parameterize the family by the degree of the polynomial.

To illustrate how we can obtain hard multivariate polynomials from hard constant-variate polynomials, suppose $g_d(x) = \sum_{i=0}^d \alpha_i x^i$ is a hard degree $d$ univariate polynomial.
We will define a new polynomial $f_n(\vec{y})$ on $n \coloneqq \floor{\log d} + 1$ variables, where the monomials of $f_n$ correspond to writing each term of $g_d$ ``in base 2.''
More precisely, for each $\vec{e} \in \bits^n$, let $j(\vec{e})$ be the number whose representation in binary corresponds to $\vec{e}$.
We assign the coefficient $\alpha_{j(\vec{e})}$ to the monomial $\vec{y}^{\vec{e}}$ in $f_n$.
To show that $f_n$ is hard, we show the contrapositive: a small circuit for $f_n$ implies a small circuit for $g_d$, which contradicts the hardness of $g_d$.
The proof of this is relatively straightforward, as we simply find a way to substitute powers of $x$ for each $y_i$ so that the monomial $\vec{y}^{\vec{e}}$ is mapped to $x^{j(\vec{e})}$.

In the case where $g_d$ is a polynomial in multiple variables, we simultaneously write each variable appearing in $g_d$ ``in base 2.''
We remark that there is nothing \textit{a priori} special about our use of base 2.
However, doing so yields polynomials which are multilinear, a fact which will be useful later on.

We now make the preceding sketch precise, showing that lower bounds in the constant-variate regime imply comparable lower bounds in the multivariate regime.

\begin{lemma} \label{lem:kronecker}
	Let $g_{m,d}(\vec{x}) = \sum_{\vec{a}} \alpha_{\vec{a}} \vec{x}^{\vec{a}}$ be a strongly $t(m,d)$-explicit $m$-variate degree $d$ polynomial which requires circuits of size $s$ to compute.
	Let $j : \bits^{\floor{\log d}+1} \to \llb 2^{\floor{\log d}+1} \rrb$ be given by $j(\vec{e}) = \sum_{i=1}^{\floor{\log d}+1} \vec{e}_i 2^{i-1}$, that is, $j(\vec{e})$ is the number whose binary representation corresponds to $\vec{e}$.
	Let $\vec{y} = (y_{1,1},\ldots,y_{1,\floor{\log d}+1},\ \ldots\ ,y_{m,1},\ldots,y_{m,\floor{\log d}+1})$ and define
	\[
		f_{m,d}(\vec{y}) = \sum_{\vec{e} \in \bits^{m \times \floor{\log d}+1}} \alpha_{(j(\vec{e}_{1,\bullet}),\ldots,j(\vec{e}_{m,\bullet}))} \vec{y}^{\vec{e}}.
	\]
	Then $f_{m,d}$ is a strongly $t(m,d)$-explicit multilinear polynomial on $m (\floor{\log d}+1)$ variables which requires circuits of size $s - \Theta(m \log d)$ to compute.
\end{lemma}

\begin{proof}
	The fact that $f_{m,d}$ is multilinear is clear from the definition.

	To see that $f_{m,d}$ is hard to compute, suppose $\Phi$ is a circuit of size $t$ which computes $f_{m,d}$.
	By applying the Kronecker substitution $y_{i,j} \mapsto x_i^{2^j}$, we can recover a circuit which computes $g_{m,d}(\vec{x})$.
	This mapping can be computed in size $\Theta(m \log d)$ by repeated squaring, so we obtain a circuit for $g_{m,d}$ of size $t + \Theta(m \log d)$.
	By assumption, $t + \Theta(m \log d) \ge s$, so $t \ge s - \Theta(m \log d)$, which proves the lower bound on the circuit complexity of $f_{m,d}$.

	Finally, remark that the binary description of a monomial in $f_{m,d}$ is exactly the same as the binary description of a monomial in $g_{m,d}$.
	This implies we can use the $t(m,d)$-time algorithm to compute the coefficients of $f_{m,d}$, so $f_{m,d}$ inherits the explicitness of $g_{m,d}$.
\end{proof}

Whether lower bounds in the multivariate regime imply lower bounds in the constant-variate regime is an open question.
In \autoref{sec:kronecker}, we give complexity-theoretic evidence that suggests the technique used to prove the preceding lemma does not suffice to prove constant-variate lower bounds from multivariate lower bounds.

In \autoref{sec:bootstrap}, we will run into some technical issues concerning circuits which are defined over a low-degree extension of the base field $\F$.
The next lemma says that whenever a circuit $\Phi$ is defined over an extension $\mathbb{K} \supseteq \F$ of low degree, such a circuit can in fact be defined over $\F$ without increasing its size too much.
A related result was proved in \textcite[\textsection 4.3]{BCS97}, where the authors considered extensions $\mathbb{K} \supseteq \F$ such that circuits defined over $\mathbb{K}$ have no computational advantage compared to circuits defined over $\F$ when computing a polynomial in $\F[\vec{x}]$.

\begin{lemma}[{\cites[Proposition 4.1(iii)]{Bur00}{HY11}}, see also {\cite[\textsection 4.3]{BCS97}}] \label{lem:simulate extension}
	Let $\F$ be a field and let $\mathbb{K} \supseteq \F$ be an extension of degree $k$.
	Suppose $f(\vec{x})$ can be computed by a circuit of size $s$ over $\mathbb{K}$.
	Then there is a circuit of size $O(k^3 s)$ which computes $f$ over $\F$.
\end{lemma}

We conclude our preliminaries on algebraic complexity by quoting a celebrated result of Kaltofen which shows that algebraic circuits may be factored without a large increase in size.

\begin{theorem}[\cite{Kaltofen89}] \label{thm:Kaltofen factoring}
	Let $f(\vec{x}) \in \F[\vec{x}]$ be a polynomial of degree $d$ computable by an algebraic circuit of size $s$.
	Let $g(\vec{x}) \in \F[\vec{x}]$ be a factor of $f(\vec{x})$.
	Then there is an algebraic circuit of size $s' \le O((snd)^4)$ which computes
	\begin{enumerate}
		\item
			$g(\vec{x})$, in the case that $\ch \F = 0$, and
		\item
			$g(\vec{x})^{p^k}$ where $k \ge 0$ is the largest integer such that $g(\vec{x})^{p^k}$ divides $f(\vec{x})$, in the case that $\ch \F = p > 0$. \qedhere
	\end{enumerate}
\end{theorem}

\subsection{Combinatorial Designs}

We will make use of the designs of \textcite{NW94}, specifically as they are used by \textcite{KI04} to prove hardness-randomness tradeoffs for polynomial identity testing.
\textcite{NW94} gave two constructions of designs: one via Reed-Solomon codes, and one via a greedy algorithm.
We first quote their construction using Reed-Solomon codes, which was also recently described in work by \textcite{KST19}.

\begin{lemma}[\cite{NW94}, see also \cite{KST19}] \label{lem:RS design}
	Let $c \ge 2$ be a positive integer, and let $n,m,\ell, r \in \naturals$ be such that (i) $\ell = m^c$, (ii) $r \le m$, (iii) $m$ is a prime power, and (iv) $n \le m^{(c-1)r}$.
	Then there is a collection of sets $S_1, \ldots, S_n \subseteq [\ell]$ such that
	\begin{itemize}
		\item
			for each $i \in [n]$, we have $|S_i| = m$; and
		\item
			for all distinct $i, j \in [n]$, we have $|S_i \cap S_j| \le r$.
	\end{itemize}
	Additionally, such a family can be deterministically constructed in $\poly(n)$ time.
\end{lemma}

We now cite the designs obtained by \textcite{NW94} via a greedy algorithm.
In the regime where $m = O(\log n)$, this improves on the previous construction by taking the size $\ell$ of the ground set to be $O(\log n)$ as opposed to $O(\log^2 n)$.

\begin{lemma}[\cite{NW94}] \label{lem:NW design}
	Let $n$ and $m$ be integers such that $n < 2^m$.
	There exists a family of sets $S_1,\ldots,S_n \subseteq [\ell]$ such that
	\begin{enumerate}
		\item
			$\ell = O(m^2 / \log(n))$,
		\item
			for each $i \in [n]$, we have $|S_i| = m$; and
		\item
			for all distinct $i, j \in [n]$, we have $|S_i \cap S_j| \le \log(n)$.
	\end{enumerate}
	Such a family of sets can be deterministically constructed in time $\poly(n,2^{\ell})$.
\end{lemma}

In extending the analysis of the Kabanets-Impagliazzo generator to low characteristic fields, we will make use of \autoref{lem:NW design}.
Our use of \autoref{lem:RS design} will arise when we combine the hardness versus randomness paradigm with the bootstrapping phenomenon.
In that setting, we will apply \autoref{lem:RS design} with $c = O(1)$ and $r = O(1)$.
Compared to \autoref{lem:NW design}, this yields sets with much smaller intersection size, though the number of sets is only $m^{O(1)}$ as opposed to $2^m$.

\subsection{Field Theory}

To cleanly state some of our results, we need the notion of a perfect field.
Namely, given a circuit $\Phi$ which computes $f(\vec{x})^p \in \F[\vec{x}]$, we will construct in \autoref{sec:pth root} a circuit $\Psi$ which computes $f(\vec{x})$.
This construction takes $p$\ts{th} roots of field elements $\alpha \in \F$, which are not always guaranteed to exist in $\F$.
To ensure $\Psi$ is defined over the base field $\F$, we require that $\F$ is closed under taking $p$\ts{th} roots, which is equivalent to requiring that $\F$ is perfect.

\begin{definition}
	A field $\F$ is called \emph{perfect} if either $\F$ has characteristic 0 or $\F$ has characteristic $p > 0$ and the map $\alpha \mapsto \alpha^p$ is an automorphism of $\F$.
	If $\F$ has characteristic $p > 0$, then the \emph{perfect closure} of $\F$, denoted $\F^{p^{-\infty}}$, is the smallest field containing $\F$ which is closed under taking $p$\ts{th} roots.
\end{definition}

It is a basic fact that perfect closures exist.

\begin{fact}
	Every field $\F$ of characteristic $p > 0$ has a perfect closure $\F^{p^{-\infty}}$.
\end{fact}

Informally, one can prove this by adjoining ``enough'' $p$\ts{th} roots to the field $\F$.
That is, for each $\alpha \in \F$, we introduce a countable collection of new field elements denoted by $(\alpha, n)$ for $n \in \naturals$, where the element $(\alpha,n)$ is meant to represent $\alpha^{p^{-n}}$.
We then take a quotient by a suitable equivalence relation; for example, if $\alpha^p = \beta$, then we regard $(\alpha,n)$ and $(\beta,n+1)$ as equivalent for all $n \in \naturals$.
One must then verify that the resulting object is in fact a field and is (up to isomorphism) the perfect closure of $\F$.
More formally, the perfect closure can be constructed as the \emph{direct limit} of a particular \emph{direct system} of fields.
We refer the reader to \textcite[Chapter 5, \textsection 1]{Bour90} for the details of this construction.

Examples of perfect fields of positive characteristic include all finite fields and all algebraically closed fields of positive characteristic.
A non-example is given by $\F_{p^m}(\vec{x})$, the field of rational functions in $n$ variables with coefficients in $\F_{p^m}$, where $\F_{p^m}$ is the finite field of size ${p^m}$.
The field $\F_{p^m}(\vec{x})$ fails to be perfect due to the fact that $x_1^{1/p} \notin \F_{p^m}(\vec{x})$, so $x_1$ is not in the image of the map $\alpha \mapsto \alpha^p$.

For more details on perfect fields, we refer the reader to any text on field theory, e.g., \textcite[Chapter 3]{Roman06}.

\section{$p$\ts{th} Roots of Algebraic Computation} \label{sec:pth root}

Suppose $\F$ is a field of characteristic $p > 0$ and $\Phi$ is a circuit which computes $f(\vec{x})^p$ for a polynomial $f(\vec{x})$.
If we want to obtain a circuit which computes $f(\vec{x})$, then \autoref{thm:Kaltofen factoring} does not suffice.
In this section, we will describe a simple transformation of $\Phi$ which yields a circuit computing $f(\vec{x})$.
This is the main technical step that will allow us to obtain hardness-randomness tradeoffs over fields of low characteristic.

In general, this transformation will incur an exponential blow-up in the size of $\Phi$.
If the original circuit computes a polynomial on $n$ variables, then the new circuit we build will be larger in size by a factor of about $p^{2n}$.
In particular, if our input is a circuit on a constant number of variables, then we only increase the size of the circuit by a constant factor.
The fact that this transformation is efficient in the constant-variate regime is exactly the reason we need to use hardness of constant-variate families of polynomials as opposed to a family of hard multilinear polynomials.

Before describing the construction for circuits on an arbitrary number of variables, we first examine the case of univariate polynomials.
Let $\F$ be a field of characteristic $p > 0$ and let $f(x) \in \F[x]$ be a univariate polynomial.
We start by grouping the monomials of $f$ by their degree modulo $p$, which allows us to write
\[
	f(x) = \sum_{i=0}^{p-1} \tilde{f}_i(x) x^i,
\]
where each $\tilde{f}_i(x)$ is a univariate polynomial in $x$ which is only supported on $p$\ts{th} powers of $x$.
That is, the term $\tilde{f}_i(x) x^i$ corresponds exactly to the monomials in $f(x)$ whose degree in $x$ is congruent to $i$ modulo $p$.
Recall that over a field of characteristic $p > 0$, we have the identity $(a + b)^p = a^p + b^p$.
Since $\tilde{f}_i(x)$ is a sum of $p$\ts{th} powers of $x$, we can write
\[
	\tilde{f}_i(x) = \sum_{j=0}^{d_i} \alpha_{i,j} x^{j p} = \del{\sum_{j=0}^{d_i} \alpha_{i,j}^{1/p} x^j}^p.
\]
This expresses $\tilde{f}_i(x)$ as a $p$\ts{th} power of the polynomial $f_i(x) \coloneqq \sum_{j=0}^{d_i} \alpha_{i,j}^{1/p} x^j$.
In general, $f_i$ may not be well-defined over $\F$, as the coefficients $\alpha_{i,j}^{1/p}$ may not exist in $\F$.
However, $\alpha_{i,j}^{1/p} \in \F^{p^{-\infty}}$, the perfect closure of $\F$, so $f_i$ is well-defined over $\F^{p^{-\infty}}$.

With this, we can write 
\[
	f(x) = \sum_{i=0}^{p-1} f_i(x)^p x^i.
\]
We refer to such an expression as the mod-$p$ decomposition of $f$.
This motivates the following definition, which generalizes this decomposition to the case of multivariate polynomials.

\begin{definition}
	Let $f(\vec{x}) \in \F[\vec{x}]$.
	The \emph{mod-$p$ decomposition of $f(\vec{x})$} is the collection of polynomials $\set{f_{\vec{a}}(\vec{x}) : \vec{a} \in \llb p \rrb^n}$ such that
	\[
		f(\vec{x}) = \sum_{\vec{a} \in \llb p \rrb^n} f_{\vec{a}}(\vec{x})^p \vec{x}^{\vec{a}}. \qedhere
	\]
\end{definition}

Over a perfect field $\F$ of characteristic $p > 0$, the existence of the mod-$p$ decomposition follows from the fact that any polynomial of the form $\sum_{\vec{a}}\alpha_{\vec{a}} \vec{x}^{p \cdot \vec{a}}$ has a $p$\ts{th} root, given by $\sum_{\vec{a}} \alpha_{\vec{a}}^{1/p}\vec{x}^{\vec{a}}$.
Here, we use the fact that $\F$ is perfect to guarantee the constants $\alpha_{\vec{a}}^{1/p}$ exist in $\F$.
Uniqueness of the decomposition follows from the fact that the monomials $\set{\vec{x}^{\vec{a}} : \vec{a} \in \naturals^n}$ form a basis for $\F[\vec{x}]$.
We record this observation as a lemma.

\begin{lemma} \label{lem:mod-p unique}
	Let $\F$ be a field of characteristic $p > 0$ and let $f,g \in \F[\vec{x}]$.
	Let $\set{f_{\vec{a}} : \vec{a} \in \llb p \rrb^n}$ and $\set{g_{\vec{a}} : \vec{a} \in \llb p \rrb^n}$ be the mod-$p$ decompositions of $f$ and $g$, respectively.
	Then $f = g$ if and only if $f_{\vec{a}} = g_{\vec{a}}$ for all $\vec{a} \in \llb p \rrb^n$.
\end{lemma}

The utility of the mod-$p$ decomposition becomes apparent when $f(\vec{x})$ is itself a $p$\ts{th} power.
In this case, $f$ itself is a sum of $p$\ts{th} powers of monomials in the variables $x_1,\ldots,x_n$, so we have $f(\vec{x}) = f_{\vec{0}}(\vec{x})^p$.
Given a circuit $\Phi$ which computes $f$, suppose we could transform $\Phi$ into a new circuit $\Psi$ which computes the mod-$p$ decomposition of $f$.
Then to compute $f(\vec{x})^{1/p}$, we simply construct the circuit $\Psi$ and set $f_{\vec{0}}(\vec{x}) = f(\vec{x})^{1/p}$ to be the output.

Before continuing on, we record a straightforward lemma about how the mod-$p$ decomposition behaves with respect to addition and multiplication.

\begin{lemma} \label{lem:mod-p decomp}
	Let $\F$ be a perfect field of characteristic $p > 0$.
	Let $f, g \in \F[\vec{x}]$, and let $\set{f_{\vec{a}} : \vec{a} \in \llb p \rrb^n}$ and $\set{g_{\vec{a}} : \vec{a} \in \llb p \rrb^n}$ be the mod-$p$ decompositions of $f$ and $g$, respectively.
	Let $h = \alpha f + \beta g$ and $q = \gamma f g$ for $\alpha,\beta,\gamma \in \F$. 
	Let $\set{h_{\vec{a}} : \vec{a} \in \llb p \rrb^n}$ and $\set{q_{\vec{a}} : \vec{a} \in \llb p \rrb^n}$ be the mod-$p$ decompositions of $h$ and $q$.
	Then for all $\vec{a} \in \llb p \rrb^n$, we have
	\[
		h_{\vec{a}} = \alpha^{1/p} f_{\vec{a}} + \beta^{1/p} g_{\vec{a}}
	\]
	and
	\[
		q_{\vec{a}} = \gamma^{1/p} \sum_{\substack{\vec{b}, \vec{c} \in \llb p \rrb^n \\ \vec{b} + \vec{c} \equiv \vec{a}\bmod{p}}} f_{\vec{b}} g_{\vec{c}} \vec{x}^{\frac{\vec{b} + \vec{c} - \vec{a}}{p}},
	\]
	where the sum and congruence $\vec{b} + \vec{c} \equiv \vec{a} \bmod{p}$ are performed component-wise.
\end{lemma}

\begin{proof}
	By expanding the equality $h = \alpha f + \beta g$ in the mod-$p$ decomposition and using the fact that $(a + b)^p = a^p + b^p$, we obtain
	\begin{align*}
		\sum_{\vec{a}\in\llb p \rrb^n} h_{\vec{a}}(\vec{x})^p \vec{x}^{\vec{a}} &= \alpha \sum_{\vec{a} \in \llb p \rrb^n} f_{\vec{a}}(\vec{x})^p \vec{x}^{\vec{a}} + \beta \sum_{\vec{a} \in \llb p \rrb^n} g_{\vec{a}}(\vec{x})^p \vec{x}^{\vec{a}} \\
		&= \sum_{\vec{a} \in \llb p \rrb^n} (\alpha^{1/p} f_{\vec{a}}(\vec{x}) + \beta^{1/p} g_{\vec{a}}(\vec{x}))^p \vec{x}^{\vec{a}}.
	\end{align*}
	\autoref{lem:mod-p unique} implies that $h_{\vec{a}} = \alpha^{1/p} f_{\vec{a}} + \beta^{1/p} g_{\vec{a}}$ as claimed.

	For $q(\vec{x})$, we again expand the equality $q = \gamma f g$ in the mod-$p$ decomposition to obtain
	\begin{align*}
		\sum_{\vec{a} \in \llb p \rrb^n} q_{\vec{a}}(\vec{x})^p \vec{x}^{\vec{a}} &= \gamma \del{\sum_{\vec{a} \in \llb p \rrb^n} f_{\vec{a}}(\vec{x})^p \vec{x}^{\vec{a}}} \del{\sum_{\vec{a} \in \llb p \rrb^n} g_{\vec{a}}(\vec{x})^p \vec{x}^{\vec{a}}} \\
		&= \gamma \sum_{\vec{b}, \vec{c} \in \llb p \rrb^n} f_{\vec{b}}(\vec{x})^p g_{\vec{c}}(\vec{x})^p \vec{x}^{\vec{b} + \vec{c}} \\
		&= \sum_{\vec{a} \in \llb p \rrb^n} \del{\gamma^{1/p} \sum_{\substack{\vec{b}, \vec{c} \in \llb p \rrb^n \\ \vec{b} + \vec{c} \equiv \vec{a} \bmod{p}}} f_{\vec{b}}(\vec{x}) g_{\vec{c}}(\vec{x}) \vec{x}^{\frac{\vec{b} + \vec{c} - \vec{a}}{p}} }^p \vec{x}^{\vec{a}}.
	\end{align*}
	Once more, \autoref{lem:mod-p unique} implies that 
	\[
		q_{\vec{a}} = \gamma^{1/p} \sum_{\substack{\vec{b}, \vec{c} \in \llb p \rrb^n \\ \vec{b} + \vec{c} \equiv \vec{a}\bmod{p}}} f_{\vec{b}} g_{\vec{c}} \vec{x}^{\frac{\vec{b} + \vec{c} - \vec{a}}{p}}
	\]
	as claimed.
\end{proof}

\subsection{Circuits} 

We start by implementing the strategy outlined above in the case of algebraic circuits.
Throughout this and subsequent sections, $\Phi$ and $\Psi$ will denote algebraic circuits, formulae, or branching programs, and $v$, $u$, and $w$ will denote gates in these circuits.
We will frequently refer to the polynomial computed at a gate $v$, which we denote by $\hat{v}$.
For $\vec{a} \in \llb p \rrb^n$, we write $\hat{v}_{\vec{a}}$ for the part of the mod-$p$ decomposition of $\hat{v}$ indexed by $\vec{a}$.

\begin{lemma}\label{lem:mod-p decomp circuit}
	Let $\F$ be a field of characteristic $p > 0$.
	Let $\Phi$ be an algebraic circuit of size $s$ which computes a polynomial $f(\vec{x}) \in \F[\vec{x}]$ and let $\set{f_{\vec{a}} : \vec{a} \in \llb p \rrb^n}$ be the mod-$p$ decomposition of $f$.
	Then there is a circuit $\Psi$ of size $3 s p^{2n} + 2^n$ which simultaneously computes $\set{f_{\vec{a}} : \vec{a} \in \llb p \rrb^n}$ over $\F^{p^{-\infty}}$, the perfect closure of $\F$.
\end{lemma}

\begin{proof}
	To construct the desired circuit $\Psi$, we will split each gate $v$ of $\Phi$ into pieces $\set{(v,\vec{a}) : \vec{a} \in \llb p \rrb^n}$ and wire $\Psi$ so that $(v,\vec{a})$ computes $\hat{v}_{\vec{a}}$.
	As $\Phi$ computes $f(\vec{x})$, this implies that $\Psi$ will contain gates computing $f_{\vec{a}}(\vec{x})$ for all $\vec{a} \in \llb p \rrb^n$.
	To wire each gate $(v,\vec{a})$ in $\Psi$, we consider the type of the gate $v$ in $\Phi$.
	\begin{itemize}
		\item
			First, suppose $v$ is an input gate in $\Phi$ labeled by a constant $\alpha \in \F$.
			In this case, we set $(v,\vec{0}) = \alpha^{1/p}$ and $(v,\vec{a}) = 0$ for $\vec{a} \neq \vec{0}$.
			By definition, $\F^{p^{-\infty}}$ contains $\alpha^{1/p}$, so this is valid over $\F^{p^{-\infty}}$.

			It follows from the definition of $\hat{v}_{\vec{a}}$ that $(v,\vec{a})$ correctly computes $\hat{v}_{\vec{a}}$.

		\item
			If $v$ is an input gate labeled by the variable $x_i$, let $\vec{e}_i$ denote the vector with a $1$ in the $i$\ts{th} slot and zero elsewhere.
			We set $(v,\vec{e}_i) = 1$ and $(v,\vec{a}) = 0$ for $\vec{a} \neq \vec{e}_i$.

			Again, it follows immediately from the definition of $\hat{v}_{\vec{a}}$ that $(v,\vec{a})$ correctly computes $\hat{v}_{\vec{a}}$.

		\item
			Suppose now that $v$ is an addition gate in $\Phi$ with children $u$ and $w$ with incoming edges labeled $\alpha_u$ and $\alpha_w$.
			For each $\vec{a} \in \llb p \rrb^p$, we set $(v,\vec{a}) = \alpha_u^{1/p} \cdot (u,\vec{a}) + \alpha_w^{1/p} \cdot (w,\vec{a})$.

			By induction, $(u,\vec{a})$ and $(w,\vec{a})$ correctly compute $\hat{u}_{\vec{a}}$ and $\hat{w}_{\vec{a}}$, respectively.
			\autoref{lem:mod-p decomp} then implies that $(v,\vec{a})$ correctly computes $\hat{v}_{\vec{a}}$.

		\item
			Finally, we consider the case where $v$ is a multiplication gate in $\Phi$ with children $u$ and $w$ with incoming edges labeled $\alpha_u$ and $\alpha_w$.
			For $\vec{a} \in \llb p \rrb^n$, we set
			\[
				(v,\vec{a}) = \alpha_u^{1/p} \alpha_w^{1/p} \sum_{\substack{\vec{b}, \vec{c} \in \llb p \rrb^n \\ \vec{b} + \vec{c} \equiv \vec{a} \pmod{p}}} (u,\vec{b}) \cdot (w,\vec{c}) \cdot \vec{x}^{\frac{\vec{b} + \vec{c} - \vec{a}}{p}},
			\]
			where vector addition and congruence of vectors is performed coordinate-wise.
			Note that since $\vec{b} + \vec{c} \equiv \vec{a} \bmod{p}$, the vector $\frac{1}{p}(\vec{b} + \vec{c} - \vec{a})$ is in fact an integer vector.
			Moreover, since $\vec{b} + \vec{c} \in \set{0,\ldots,2(p-1)}^n$, it follows that $\vec{b} + \vec{c} - \vec{a} \in \set{0,p}^n$, so $\frac{1}{p}(\vec{b} + \vec{c} - \vec{a}) \in \bits^n$ is a zero-one vector.
			
			Via induction, $(u,\vec{b})$ and $(w,\vec{c})$ correctly compute $\hat{u}_{\vec{b}}$ and $\hat{w}_{\vec{c}}$, respectively.
			From this and \autoref{lem:mod-p decomp}, it follows that $(v,\vec{a})$ correctly computes $\hat{v}_{\vec{a}}$.
	\end{itemize}

	As previously remarked, since $\Phi$ computes $f(\vec{x})$, for every $\vec{a} \in \llb p \rrb^n$ there is a gate in $\Psi$ which computes $f_{\vec{a}}(\vec{x})$, so $\Psi$ correctly computes all components of the mod-$p$ decomposition of $f$.
	It remains to bound the size of $\Psi$.

	For every gate in $\Phi$, we construct $p^n$ gates of the form $(v,\vec{a})$ in $\Psi$.
	In the case that $v$ is a multiplication gate, we need extra intermediate hardware to compute the summation $(v,\vec{a}) = \sum_{\vec{b} + \vec{c} \equiv \vec{a} \pmod{p}} (u,\vec{b}) \cdot (w,\vec{c}) \cdot \vec{x}^{\frac{\vec{b} + \vec{c} - \vec{a}}{p}}$.
	This can be done with $p^n$ summation gates and $2 p^n$ multiplication gates.
	We also need $2^n$ gates to compute the products $\vec{x}^{\vec{e}}$ for $\vec{e} \in \bits^n$.
	Since $\Psi$ is a circuit, we only need to pay for these gates once, as we can reuse them for all the multiplication computations.
	In total, each multiplication gate incurs an extra cost of $3 p^n$ gates.

	This implies each gate in $\Phi$ gives rise to at most $3 p^{2n}$ gates in $\Psi$.
	As there are $s$ gates in $\Phi$, there are at most $3 s p^{2n} + 2^n$ gates in $\Psi$.
\end{proof}

\begin{remark}
	In the above construction, rather than using the perfect closure, the resulting circuit can be defined over an extension $\mathbb{K} \supseteq \F$ of finite degree.
	This can be done by adjoining to $\F$ all $p$\ts{th} roots of constants which appear in $\Phi$.
	The degree of this extension may be exponential in $s$ in the worst case.
\end{remark}

We can now use the construction of \autoref{lem:mod-p decomp circuit} to take $p$\ts{th} roots of circuits which compute a $p$\ts{th} power over a field of characteristic $p$.

\begin{corollary}\label{cor:pth root of circuit}
	Let $\F$ be a field of characteristic $p > 0$.
	Let $\Phi$ be an algebraic circuit of size $s$ which computes a polynomial $f(\vec{x})^p \in \F[\vec{x}]$.
	Then there is a circuit $\Psi$ of size $3 s p^{2n} + 2^n$ which computes $f(\vec{x})$ over $\F^{p^{-\infty}}$, the perfect closure of $\F$.
\end{corollary}

\begin{proof}
	By \autoref{lem:mod-p decomp circuit}, there is a circuit $\Psi$ of the claimed size which computes $(f(\vec{x})^p)_{\vec{0}}$.
	It follows from the definition of the mod-$p$ decomposition that $f(\vec{x}) = (f(\vec{x})^p)_{\vec{0}}$, so $\Psi$ computes $f(\vec{x})$ as desired.
\end{proof}

\begin{remark}
	If $n = O(\log_p s)$, then \autoref{cor:pth root of circuit} shows that if $f^p$ is computable in size $s$, then $f$ is computable in size $s^{O(1)}$.
	While the log-variate regime may appear as a somewhat artificial intermediary between the constant-variate and full multivariate regimes, it is a meaningful setting to study due to various corollaries of the bootstrapping results.
	For example, \textcite{FGS18} recently studied the problem of designing explicit hitting sets for log-variate depth-three diagonal circuits.
\end{remark}

\subsection{Formulae}

It is natural to ask if the mod-$p$ decomposition allows us to efficiently take $p$\ts{th} roots in other models of algebraic computation.
We address this question first in the case of algebraic formulae, and subsequently for algebraic branching programs.
For the reader who is solely interested in the application of the mod-$p$ decomposition and \autoref{cor:pth root of circuit} to hardness-randomness tradeoffs, it is safe to skip ahead to \autoref{sec:KI}.
Before continuing on, we make an important remark regarding formulae and branching programs for univariate polynomials.

\begin{remark}
	In the univariate regime, our results (as stated) for formulae and branching programs are not as meaningful as the result for circuits.
	A formula or ABP of size $s$ can only compute a polynomial of degree $d \le s$, so any formula or ABP computing a degree $d$ univariate polynomial must have size at least $d$.
	For univariate polynomials, Horner's rule supplies a matching $O(d)$ upper bound.
	Thus, the $p$\ts{th} root of a univariate polynomial which has complexity $s$ can be computed by a device of size $s/p$, which is much stronger than what we will obtain in \autoref{cor:pth root of formula} and \autoref{cor:pth root of abp}.

	However, if one modifies the model of formulae (or branching programs) to allow leaves (or edges) labeled by a power of a variable $x_i^j$, then the trivial $\Omega(d)$ lower bound no longer holds.
	Our techniques can be adapted to this stronger model with little modification, where the upper bounds we obtain are less trivial.
\end{remark}

We now show how one can compute the mod-$p$ decomposition of an algebraic formula.
We essentially do this by applying the transformation of \autoref{lem:mod-p decomp circuit} and arguing that we can convert the resulting circuit into a formula without increasing its size too much.
To do this, we need some additional bookkeeping to ensure that the underlying graph of the resulting computation is a tree.
We borrow this style of bookkeeping from \textcite{Raz13}, who used it for improved homogenization and multilinearization of formulae.
Alternatively, one can use the fact that formulae of size $s$ can be rebalanced to have depth $O(\log s)$ and then analyze the increase in depth incurred in the proof of \autoref{lem:mod-p decomp circuit}.

\begin{lemma} \label{lem:mod-p decomp formula}
	Let $\F$ be a field of characteristic $p > 0$.
	Let $\Phi$ be an algebraic formula of size $s$ and product depth $d$ which computes a polynomial $f(\vec{x}) \in \F[\vec{x}]$ and let $\set{f_{\vec{a}} : \vec{a} \in \llb p \rrb^n}$ be the mod-$p$ decomposition of $f$.
	Then there is a formula $\Psi$ of size $3 s n p^{n (d+3)}$ and product depth $d + \ceil{\log n}$ which simultaneously computes $\set{f_{\vec{a}} : \vec{a} \in \llb p \rrb^n}$ over $\F^{p^{-\infty}}$, the perfect closure of $\F$.
\end{lemma}

\begin{proof}
	As in \autoref{lem:mod-p decomp circuit}, we will split each gate $v$ of $\Phi$ into pieces which compute components of the mod-$p$ decomposition of $\hat{v}$.
	However, we will need a much larger number of copies of $v$ to ensure that the resulting circuit $\Psi$ is in fact a formula.

	We first set up some notation, borrowing heavily from \textcite{Raz13}.
	For a gate $v$ in $\Phi$, let $\pathrm(v)$ denote the set of all vertices on the path from $v$ to the root of $\Phi$, including $v$ itself.
	Let $N_v$ denote the set of all functions $T : \pathrm(v) \to \llb p \rrb^n$ such that for all $u, w \in \pathrm(v)$ where $u$ is a sum gate with child $w$, we have $T(u) = T(w)$.
	Informally, the map $T$ encodes the progression of types in the mod-$p$ decomposition seen as the computation progresses through the formula.

	For each gate $v$ in $\Phi$, we create a collection of gates $\set{(v,\vec{a},T) : \vec{a}\in\llb p \rrb^n, T \in N_v, T(v) = \vec{a}}$.
	We will wire the gates of $\Psi$ so that $(v,\vec{a},T)$ computes $\hat{v}_{\vec{a}}$.
	As before, to wire the gates of $\Psi$ correctly, we consider what type of gate $v$ is in $\Phi$.
	The construction only differs meaningfully from that of \autoref{lem:mod-p decomp circuit} in the case of multiplication gates.
	\begin{itemize}
		\item
			If $v$ is an input gate in $\Phi$ labeled by $\alpha \in \F$, then we set $(v,\vec{0},T) = \alpha^{1/p}$ and $(v,\vec{a},T) = 0$ for $\vec{a} \neq \vec{0}$.
			As $\alpha^{1/p} \in \F^{p^{-\infty}}$, this produces a valid circuit over $\F^{p^{-\infty}}$.

			It is immediate from the definition that $(v,\vec{a},T)$ correctly computes $\hat{v}_{\vec{a}}$.

		\item
			If $v$ is an input gate labeled by the variable $x_i$, let $\vec{e}_i$ denote the vector with a $1$ in the $i$\ts{th} slot and zero elsewhere.
			We set $(v,\vec{e}_i,T) = 1$ and $(v,\vec{a},T) = 0$ for $\vec{a} \neq \vec{e}_i$.

			Once more, it is an immediate consequence of the definition that $(v,\vec{a},T)$ correctly computes $\hat{v}_{\vec{a}}$.

		\item
			Suppose now that $v$ is an addition gate with children $u$ and $w$ with incoming edges labeled $\alpha_u$ and $\alpha_w$.
			For each $\vec{a} \in \set{0,\ldots,p-1}^n$ and $T \in N_v$, we set $(v,\vec{a},T) = \alpha_u^{1/p} \cdot (u,\vec{a}, T_u) + \alpha_w^{1/p} \cdot (w, \vec{a}, T_w)$, where $T_u \in N_u$ and $T_w \in N_w$ extend $T$ and satisfy $T(v) = T_u(u) = T_w(w)$.

			By induction, $(u,\vec{a},T_u)$ and $(w,\vec{a},T_w)$ correctly compute $\hat{u}_{\vec{a}}$ and $\hat{w}_{\vec{a}}$, respectively.
			By \autoref{lem:mod-p decomp}, it follows that $(v,\vec{a},T)$ correctly computes $\hat{v}_{\vec{a}}$.

		\item
			Finally, consider the case when $v$ is a multiplication gate with children $u$ and $w$ with incoming edges labeled $\alpha_u$ and $\alpha_w$.
			We set
			\[
				(v,\vec{a},T) = \alpha_u^{1/p} \alpha_w^{1/p} \sum_{\vec{b} + \vec{c} \equiv \vec{a}\pmod{p}} (u, \vec{b}, T_{u,\vec{b}}) \cdot (w, \vec{c}, T_{w,\vec{c}}) \cdot \vec{x}^{\frac{\vec{b} + \vec{c} - \vec{a}}{p}},
			\]
			where $T_{u,\vec{b}}$ (respectively $T_{w,\vec{c}}$) extends $T$ and satisfies $T_{u,\vec{b}}(u) = \vec{b}$ (respectively $T_{w,\vec{c}}(w) = \vec{c}$).

			By induction, $(u,\vec{b},T_{u,\vec{b}})$ and $(w,\vec{c},T_{w,\vec{c}})$ compute $\hat{u}_{\vec{b}}$ and $\hat{w}_{\vec{c}}$, respectively.
			\autoref{lem:mod-p decomp} implies that $(v,\vec{a},T)$ correctly computes $\hat{v}_{\vec{a}}$.
	\end{itemize}

	By construction, $\Psi$ correctly computes $\set{f_{\vec{a}} : \vec{a} \in \llb p \rrb^n}$.
	It remains to bound the size and product depth of $\Psi$ and show that $\Psi$ is indeed a formula.

	Each gate $v$ in $\Phi$ yields $p^n |N_v|$ gates of the form $(v,\vec{a},T)$ in $\Psi$.
	If $v$ is a multiplication gate with children $u$ and $w$, we need to implement the sum over the children $(u,\vec{b},T_u)$ and $(w,\vec{c},T_w)$.
	For a given $\vec{e} \in \bits^n$, we can compute $\vec{x}^{\vec{e}}$ using a subformula of size at most $n$.
	To compute $(v,\vec{a},T)$, we need $p^n$ summation gates and $2 p^n$ multiplication gates in addition to the gates computing $(u,\vec{b},T_u)$, $(w,\vec{c},T_w)$, and $\vec{x}^{\vec{e}}$.
	This implies that we can compute $(v,\vec{a},T)$ using at most $3 n p^n$ extra gates.
	Thus, for every gate $v$ in $\Phi$, we create at most $3n p^{2n} |N_v|$ gates in $\Psi$.

	To bound the size of $N_v$, note that a function $T \in N_v$ can only change values along $\pathrm(v)$ at multiplication gates.
	Since there are at most $d$ multiplication gates along $\pathrm(v)$, we can specify $T$ by a $(d+1)$-tuple of elements of $\llb p \rrb^n$, corresponding to the values taken by $T$ between successive multiplication gates.
	This implies $|N_v| \le p^{n(d+1)}$.
	Thus $\Psi$ contains at most $3 s n p^{n(d+3)}$ gates.

	It follows from the definition of $\Psi$ that the product depth of $\Psi$ is $d+\ceil{\log n}$, as the number of product gates on any path from a leaf to the root increases by at most an additive $\ceil{\log n}$.
	This arises from the need to implement a product of the form $\vec{x}^{\vec{e}}$ at gates of $\Psi$ which correspond to multiplication gates in $\Phi$.
	As we need to compute a product of this form at most once along every path from the root to a leaf, we only incur an additive $\ceil{\log n}$ increase in product depth as opposed to a multiplicative increase.

	To see that $\Psi$ is a formula, consider the edges leaving the gate $(u,\vec{a},T)$.
	Let $v$ denote the parent of $u$ in $\Psi$.
	If $v$ is an addition gate, then only $(v,\vec{a},T_v)$ receives an edge from $(u,\vec{a},T)$ where $T_v \in N_v$ agrees with $T$ on $\pathrm(v)$.
	If $v$ is a multiplication gate, then only $(v,T(v),T_v)$ receives an edge from $(u,\vec{a},T)$ where $T_v \in N_v$ agrees with $T$ on $\pathrm(v)$.
	In both cases, the fan-out of the gate $u$ is $1$, so $\Psi$ is in fact a formula.
\end{proof}

As with circuits, we can use \autoref{lem:mod-p decomp formula} to compute $p$\ts{th} roots of formulae which compute a $p$\ts{th} power over a field of characteristic $p > 0$.

\begin{corollary} \label{cor:pth root of formula}
	Let $\F$ be a field of characteristic $p > 0$.
	Let $\Phi$ be an algebraic formula of size $s$ and product depth $d$ which computes a polynomial $f(\vec{x})^p \in \F[\vec{x}]$.
	Then there is a formula $\Psi$ of size $3 s n p^{n(d + 3)}$ and product depth $d + \ceil{\log n}$ which computes $f(\vec{x})$ over $\F^{p^{-\infty}}$, the perfect closure of $\F$.
\end{corollary}

\begin{proof}
	Analogous to the proof of \autoref{cor:pth root of circuit}.
\end{proof}

\subsection{Algebraic Branching Programs}

We now consider the task of taking $p$\ts{th} roots of algebraic branching programs.
We consider the model of branching programs where edges may only be labeled by a constant $\alpha \in \F$ or a multiple of a variable $\alpha x_i$.
Some authors allow the edges of a branching program to be labeled by an affine form $\ell(\vec{x}) = \alpha_0 + \sum_{i=1}^n \alpha_i x_i$.
Such a branching program can be converted to one whose edges are labeled by field constants or multiples of a variable.
This transformation increases the number of vertices by a factor of $O(n)$, which is small compared to the increase in size we will incur by taking a $p$\ts{th} root.
We begin by computing the mod-$p$ decomposition of an algebraic branching program.

\begin{lemma} \label{lem:mod-p decomp abp}
	Let $\F$ be a field of characteristic $p > 0$.
	Let $\Phi$ be an algebraic branching program on $s$ vertices with edges labeled by variables or field constants which computes a polynomial $f(\vec{x}) \in \F[\vec{x}]$ and let $\set{f_{\vec{a}} : \vec{a} \in \llb p \rrb^n}$ be the mod-$p$ decomposition of $f$.
	Then there is an algebraic branching program $\Psi$ on $s p^{n}$ vertices which simultaneously computes $\set{f_{\vec{a}} : \vec{a} \in \llb p \rrb^n}$ over $\F^{p^{-\infty}}$, the perfect closure of $\F$.
\end{lemma}

\begin{proof}
	For each node $v$ in $\Phi$, we create a collection of nodes $\set{(v,\vec{a}) : \vec{a} \in \llb p \rrb^n}$ in $\Psi$.
	We will wire the nodes of $\Psi$ so that $(v,\vec{a})$ computes $\hat{v}_{\vec{a}}$.

	For a pair of vertices $u$ and $v$, let $\ell(u,v)$ denote the label of the edge between $u$ and $v$.
	Let $N^{\mathrm{in}}(v)$ denote the set of vertices $w$ such that the edge $(w,v)$ is present in $\Phi$.

	Let $u$ and $v$ be two nodes in $\Phi$ and suppose there is an edge from $u$ to $v$ in $\Phi$.
	We consider two cases, depending on whether this edge is labeled by a constant $\alpha \in \F$ or a multiple of a variable $\alpha x_i$.
	\begin{itemize}
		\item
			Suppose the edge from $u$ to $v$ is labeled by $\alpha \in \F$.
			For all $\vec{a} \in \llb p \rrb^n$, we add an edge between $(u,\vec{a})$ and $(v, \vec{a})$ labeled by $\alpha^{1/p}$.
			Since $\alpha^{1/p} \in \F^{p^{-\infty}}$, this construction is valid over the perfect closure $\F^{p^{-\infty}}$ of $\F$.

		\item
			Suppose the edge from $u$ to $v$ is labeled by $\alpha x_i$, where $\alpha \in \F$.
			Denote by $\vec{e}_i$ the vector which has a $1$ in the $i$\ts{th} slot and zeroes elsewhere.
			For all $\vec{a} \in \llb p \rrb^n$, we add an edge between $(u,\vec{a})$ and $(v, \vec{a} + \vec{e}_i)$, where the addition $\vec{a} + \vec{e}_i$ is performed modulo $p$.
			If $\vec{a}_i < p-1$, we label this edge with $\alpha^{1/p}$.
			If $\vec{a}_i = p-1$, we label this edge with $\alpha^{1/p} x_i$.
			Again, $\alpha^{1/p} \in \F^{p^{-\infty}}$ by definition, so this construction is valid.
	\end{itemize}

	To see that this construction is correct, let $v$ be a node in $\Phi$.
	By the definition of an algebraic branching program, we have
	\[
		\hat{v} = \sum_{u \in N^{\mathrm{in}}(v)} \ell(u,v) \cdot \hat{u}.
	\]
	Repeatedly applying the addition case of \autoref{lem:mod-p decomp} yields, for each $\vec{a} \in \llb p \rrb^n$,
	\[
		\hat{v}_{\vec{a}} = \sum_{u \in N^{\mathrm{in}}(v)} (\ell(u,v) \cdot \hat{u})_{\vec{a}}.
	\]
	If $\ell(u, v) = \alpha \in \F$, then we have $(\ell(u, v) \cdot \hat{u})_{\vec{a}} = \alpha^{1/p} \hat{u}_{\vec{a}}$.
	If $\ell(u, v) = \alpha x_i$, then if $\vec{a}_i > 0$, we have $(\ell(u, v) \cdot \hat{u})_{\vec{a}} = \alpha^{1/p} \hat{u}_{\vec{a} - \vec{e}_i}$.
	Otherwise, $\vec{a}_i = 0$, so $(\ell(u,v) \cdot \hat{u})_{\vec{a}} = \alpha^{1/p} \hat{u}_{\vec{a} - \vec{e}_i} x_i$, where the subtraction $\vec{a} - \vec{e}_i$ is done modulo $p$.

	By induction, $(u,\vec{a})$ correctly computes $\hat{u}_{\vec{a}}$.
	From our construction of $\Psi$, if $(u,v)$ is an edge in $\Phi$, then $(v,\vec{a})$ has an incoming edge which computes $(\ell(u, v) \cdot \hat{u})_{\vec{a}}$.
	This implies that $(v,\vec{a})$ computes the polynomial $\sum_{u \in N^{\mathrm{in}}(v)} (\ell(u,v) \cdot \hat{u})_{\vec{a}} = \hat{v}_{\vec{a}}$, which is what we want.

	Thus, $\Psi$ simultaneously computes $\set{f_{\vec{a}} : \vec{a} \in \llb p \rrb^n}$.
	Every node in $\Phi$ corresponds to $p^n$ nodes in $\Psi$.
	Unlike the cases of circuits and formulae, we do not need extra hardware to implement intermediate calculations, so $\Psi$ consists of $s p^n$ nodes as claimed.
\end{proof}

Again, as in the case of circuits and formulae, this immediately yields a way to compute $p$\ts{th} roots of algebraic branching programs which compute a $p$\ts{th} power over a field of characteristic $p > 0$.

\begin{corollary} \label{cor:pth root of abp}
	Let $\F$ be a field of characteristic $p > 0$.
	Let $\Phi$ be an algebraic branching program on $s$ vertices with edges labeled by variables or field constants which computes a polynomial $f(\vec{x})^p \in \F[\vec{x}]$.
	Then there is an algebraic branching program $\Psi$ on $s p^{n}$ vertices which computes $f(\vec{x})$ over $\F^{p^{-\infty}}$, the perfect closure of $\F$.
\end{corollary}

\begin{proof}
	Analogous to the proof of \autoref{cor:pth root of circuit}.
\end{proof}

\section{Extending the Kabanets-Impagliazzo Generator} \label{sec:KI}

With our main technical tool in hand, we move on to our first application.
The hitting set generator of \textcite{KI04} was the first to provide hardness-randomness tradeoffs for polynomial identity testing over fields of characteristic zero.
Over fields of characteristic $p > 0$, \citeauthor{KI04} obtain hardness-randomness tradeoffs under non-standard hardness assumptions.
Namely, they require an explicit family of polynomials $\set{f_n : n \in \naturals}$ such that $f_n^{p^k}$ is hard to compute for $1 \le p^k \le 2^{O(n)}$, though they do not state their results in this way.
Rather, they use the assumption of a family of polynomials which are hard to compute as functions, which implies hardness of $p$\ts{th} powers over finite fields.

It is more common in algebraic complexity to prove lower bounds on the task of computing polynomials as syntactic objects.
Over infinite fields, this is equivalent to computing a polynomial as a function.
However, the two notions differ over finite fields.
For example, the polynomial $x^2-x$ is non-zero as a polynomial over $\F_2$, but computes the zero function over $\F_2$.
It is interesting to note that examples of functional lower bounds over finite fields are known.
The works of \textcite{GK98,GR00,KS17} prove lower bounds against constant-depth circuits over finite fields which functionally compute an explicit polynomial.

In this section, we will extend the Kabanets-Impagliazzo generator to all perfect fields of characteristic $p > 0$ under syntactic hardness assumptions for a single family of polynomials.
The perfect fields of characteristic $p$ include all finite fields and all algebraically closed fields of positive characteristic.
To do this, we need a stronger (but still syntactic) hardness assumption.
In their work, \citeauthor{KI04} use the existence of an explicit family of hard multilinear polynomials to derandomize polynomial identity testing.
Here, we need lower bounds against an explicit family of constant-variate polynomials of arbitrarily high degree.
Such an assumption appears to be stronger than the assumption of a hard family of multilinear polynomials.
We discuss the relationship between these hypotheses in more detail in \autoref{sec:kronecker}.

\subsection{The Kabanets-Impagliazzo Generator}

We first describe the construction of the Kabanets-Impagliazzo generator.

\begin{construction}[\cite{KI04}] \label{cons:KI generator}
	Let $n$ and $m$ be integers satisfying $n < 2^m$.
	Let $g \in \F[\vec{x}]$ be a polynomial on $m$ variables.
	Let $S_1,\ldots,S_n \subseteq [\ell]$ be a Nisan-Wigderson design as in \autoref{lem:NW design}.
	The Kabanets-Impagliazzo generator $\mathcal{G}_{\mathrm{KI},g}(\vec{z}) : \F^\ell \to \F^n$ is the polynomial map given by
	\[
		\mathcal{G}_{\mathrm{KI},g}(\vec{z}) \coloneqq (g(\vec{z}|_{S_1}),\ldots,g(\vec{z}|_{S_n})),
	\]
	where $\vec{z}|_{S_i}$ denotes the restriction of $\vec{z}$ to the variables with indices in $S_i$.
\end{construction}

We now quote the main lemma used by \citeauthor{KI04} in the analysis of their generator.

\begin{lemma}[\cite{KI04}] \label{lem:KI lemma}
	Let $\F$ be any field and $n, m \in \naturals$ such that $n < 2^m$.
	Let $f \in \F[y_1,\ldots,y_n]$ and $g \in \F[x_1,\ldots,x_m]$ be non-zero polynomials of degree $d_f$ and $d_g$, respectively.
	Let $f(\vec{y})$ be computable by an algebraic circuit of size $s$.
	Let $S \subseteq \F$ be any set of size at least $d_f d_g + 1$ and let $\ell = O(m^2 / \log n)$ be as in \autoref{lem:NW design}.
	Let $\mathcal{G}_{\mathrm{KI},g}$ be as in \autoref{cons:KI generator}.

	Suppose that $f(\mathcal{G}_{\mathrm{KI},g}(\vec{\alpha})) = 0$ for all $\vec{\alpha} \in S^{\ell}$.
	Then there is an algebraic circuit $\Phi$ of size $s' \le \poly(n, m, d_f, d_g, s, (1+\ideg g)^{\log n})$ which computes the following.
	If $\F$ has characteristic zero, then $\Phi$ computes $g(\vec{x})$.
	If $\F$ has characteristic $p > 0$, then $\Phi$ computes $g(\vec{x})^{p^k}$ for some $k \in \naturals$ such that $p^k \le d_f$.
\end{lemma}

If $f(\mathcal{G}_{\mathrm{KI},g}(\vec{z})) = 0$, then using \autoref{lem:KI lemma}, we can reconstruct a circuit for $g$ using the circuit for $f$.
By taking $g$ from a family of hard polynomials, we obtain a contradiction if there is a small circuit which computes $f$.
This proves that $\mathcal{G}_{\mathrm{KI},g}$ is a hitting set generator for the class of small circuits.
The explicitness of $\mathcal{G}_{\mathrm{KI},g}$ follows from the explicitness of the family from which $g$ is taken.
The hardness-randomness tradeoffs of \textcite{KI04} then follow by setting parameters according to the hardness of $g$.

Over a field of characteristic $p>0$, \autoref{lem:KI lemma} provides a circuit computing $g(\vec{x})^{p^k}$.
Suppose we are working over $\F_q$, the finite field of $q = p^a$ elements.
By taking $p$\ts{th} powers of $g(\vec{x})^{p^k}$ if necessary, we can obtain a circuit which computes $g(\vec{x})^{p^{a r}} = g(\vec{x})^{q^r}$ for some $r \in \naturals$.
The map $\alpha \mapsto \alpha^q$ is the identity over $\F_q$, so the circuit which computes $g(\vec{x})^{q^r}$ in fact computes the same function as $g(\vec{x})$.
This is why, without further work, we need a polynomial which is hard to compute as a function to obtain hardness-randomness tradeoffs over finite fields.

If we could factor the circuit for $g(\vec{x})^{p^k}$ to obtain a not-too-much-larger circuit for $g(\vec{x})$, then we could derive hardness-randomness tradeoffs from the assumption of an explicit family of multilinear polynomials which are hard to compute.
It remains an open problem to show that if $g(\vec{x})^p$ has a small circuit, then $g(\vec{x})$ has a small circuit.
However, in the constant-variate regime, \autoref{cor:pth root of circuit} resolves this problem in the affirmative.
This is the main fact which drives our extension of the Kabanets-Impagliazzo generator.

\subsection{Extension to Fields of Low Characteristic}

We now show how to use the Kabanets-Impagliazzo generator to obtain hardness-randomness tradeoffs over all perfect fields of characteristic $p > 0$. 
Recall that $\mathcal{C}_{\F}(s,n,d)$ denotes the set of $n$-variate degree $d$ polynomials computable by circuits of size at most $s$.

\begin{theorem} \label{thm:KI hard vs rand}
	Let $\F$ be a field of characteristic $p > 0$ and let $c, k \in \naturals$ be positive constants.
	Let $\set{g_{d}(\vec{x}) : d \in \naturals}$ be a strongly $t(k,d)$-explicit family of $k$-variate degree $d$ polynomials.
	Let $s : \naturals \to \naturals$ be a function such that $g_d$ cannot be computed by algebraic circuits of size smaller than $s(d)$ over $\F^{p^{-\infty}}$.
	Then there is a hitting set generator $\mathcal{G} : \F^{\ell} \to \F^n$ for $\mathcal{C}_{\F}(n^c,n,n^c)$ which
	\begin{enumerate}
		\item
			is $\del{\poly(n,2^{\ell}) + t(k,n^{3ck + \Omega(c)}) \cdot s^{-1}(n^{3ck + \Omega(c)})^{O(k)}}$-explicit,
		\item
			has seed length $\ell = O\del{\frac{k^2 \log^2(s^{-1}(n^{3ck + O(c)}))}{\log n}}$, and
		\item
			has degree $O(k \log(s^{-1}(n^{3ck + O(c)})))$.
	\end{enumerate}
\end{theorem}

\begin{proof}
	We will obtain our generator by using $\set{g_d : d \in \naturals}$ to construct a family of hard multilinear polynomials.
	We then set parameters and instantiate the Kabanets-Impagliazzo generator with this hard multilinear family.

	By \autoref{lem:kronecker}, there is a strongly $t(k,d)$-explicit family of multilinear polynomials $h_d(\vec{y})$ on $m \coloneqq k (\floor{\log d} + 1)$ variables such that any circuit which computes $h_d$ must be of size $s(d) - O(k \log d)$.
	The construction of $h_d$ also yields the identity
	\[
		g_d(\vec{x}) = h_d(x_1^{2^0},x_1^{2^1},\ldots,x_1^{2^{\floor{\log d}}},\ \ldots\ ,x_k^{2^0},x_k^{2^1},\ldots,x_k^{2^{\floor{\log d}}}),
	\]
	which allows us to obtain a circuit for $g_d$ from a circuit for $h_d$.
	As $h_d$ is multilinear, we have $\deg(h_d) \le m$ and $\ideg(h_d) = 1$.

	Set $d = s^{-1}(n^e)$ for a large enough constant $e \ge 1$ to be specified later.
	Since $g_d$ is a $k$-variate degree $d$ polynomial, we trivially have $s(d) \le d^{O(k)}$, so $s^{-1}(d) \ge d^{\Omega(1/k)}$.
	This gives us
	\[
		2^m \ge d^k = s^{-1}(n^e)^k \ge (n^{\Omega(e/k)})^k = n^{\Omega(e)}.
	\]
	Taking $e$ to be large enough guarantees $2^m > n$.
	Let $S_1,\ldots,S_n \subseteq [\ell]$ be the Nisan-Wigderson design guaranteed by \autoref{lem:NW design}.
	Our generator $\mathcal{G} : \F^{\ell} \to \F^n$ is given by instantiating the Kabanets-Impagliazzo generator with $h_d$.
	That is,
	\[
		\mathcal{G}(\vec{z}) \coloneqq \mathcal{G}_{\mathrm{KI},{h_d}}(\vec{z}) = (h_d(\vec{z}|_{S_1}), \ldots, h_d(\vec{z}|_{S_n})).
	\]
	We now verify the claimed properties of $\mathcal{G}$.

	\textbf{Correctness.}
	To see that $\mathcal{G}$ is indeed a hitting set generator for $\mathcal{C}_{\F}(n^c,n,n^c)$, suppose there is some non-zero $f \in \mathcal{C}_{\F}(n^c, n, n^c)$ such that $f(\mathcal{G}(\vec{z})) = 0$.
	Then by \autoref{lem:KI lemma}, there is a circuit of size 
	\[
		s' \le \poly(n,m,n^c,2^{\log n}) \le n^{O(c)}
	\]
	which computes $h_d(\vec{y})^{p^a}$ for $p^a \le \deg(f) \le n^c$.
	Via the Kronecker substitution $y_{i,j} \mapsto x_i^{2^j}$, we obtain a circuit of size $s' + O(k \log d) \le n^{O(c)}$ which computes $g_d(\vec{x})^{p^a}$.
	We now apply \autoref{cor:pth root of circuit} a total of $a$ times to obtain a circuit which computes $g_d(\vec{x})$ and has size $s'' \le (3 \cdot 2^k \cdot p^{2k})^a n^{O(c)}$.
	Since $p^a \le n^c$ and $2 \le p$, we obtain $s'' \le n^{3kc + O(c)}$.
	By setting $e = 3ck + \Theta(c)$ where the hidden constant on the $\Theta(c)$ term is large enough, we obtain a contradiction as follows.
	By assumption, any circuit which computes $g_d$ must be of size at least $s(d) = n^e$.
	However, we have a circuit of size $n^{3ck + O(c)} \ll n^e = s(d)$ which computes $g_d$, a contradiction.
	Thus, it must be the case that $f(\mathcal{G}(\vec{z})) \neq 0$.
	Hence $\mathcal{G}$ is a hitting set generator for $\mathcal{C}_{\F}(n^c,n,n^c)$.

	\textbf{Explicitness.}
	Given a point $\vec{\alpha} \in \F^{\ell}$, we can evaluate $\mathcal{G}$ as follows.
	First, we construct the Nisan-Wigderson design $S_1,\ldots,S_n \subseteq [\ell]$ in time $\poly(n,2^\ell)$.
	We then compute all $d^{O(k)}$ coefficients of $h_d$, each in $t(k,d)$ time.
	Finally, for each $i \in [\ell]$, we evaluate $h_d$ on $\vec{\alpha}|_{S_i}$ in time $d^{O(k)}$.
	Using the fact that $d = s^{-1}(n^{3ck + O(c)})$, we can evaluate $\mathcal{G}$ in $\poly(n,2^\ell) + t(k,n^{3ck + O(c)}) \cdot s^{-1}(n^{3ck + O(c)})^{O(k)}$ time as claimed.

	\textbf{Seed length.}
	It follows from \autoref{lem:NW design} that $\mathcal{G}$ has seed length $\ell = O(m^2/\log n)$ = $O\del{\frac{k^2 \log^2 d}{\log n}}$.
	By our choice of $d = s^{-1}(n^{3ck + O(c)})$, we obtain the claimed seed length of $O\del{\frac{k^2 \log^2(s^{-1}(n^{3ck + O(c)}))}{\log n}}$.

	\textbf{Degree.}
	By construction, $\mathcal{G}$ is a map of degree $\deg(h_d) \le m = k (\floor{\log d}+1)$.
	Once more, plugging in our choice of $d$ yields the claimed bound of $O(k \log(s^{-1}(n^{3ck + O(c)})))$.
\end{proof}

By applying \autoref{lem:hsg to hitting set}, we obtain the following construction of explicit hitting sets for $\mathcal{C}_{\F}(n^c,n,n^c)$.

\begin{corollary} \label{cor:KI hitting set}
	Assume the setup of \autoref{thm:KI hard vs rand}. 
	Let $T$, $\ell$, and $\Delta$ be the explicitness, seed length, and degree of the generator of \autoref{thm:KI hard vs rand}, respectively.
	Then there is a hitting set $\mathcal{H}$ for $\mathcal{C}_{\F}(n^c,n,n^c)$ which
	\begin{enumerate}
		\item
			has size $|\mathcal{H}| = (n^c \Delta + 1)^\ell$, and 
		\item
			has explicitness $|\mathcal{H}| \cdot T = (n^c \Delta + 1)^\ell \cdot T$.
	\end{enumerate}
\end{corollary}

\begin{proof}
	This is \autoref{lem:hsg to hitting set} applied to \autoref{thm:KI hard vs rand}.
\end{proof}

We conclude this section with some concrete hardness-randomness tradeoffs obtainable via \autoref{thm:KI hard vs rand} and \autoref{cor:KI hitting set}.
Recall that for constant $k$, a $k$-variate polynomial of degree $d$ consists of at most $\binom{k + d}{k} \le d^{O(k)}$ monomials.
In this regime, a polynomial which is strongly $d^{O(k)}$-explicit is ``exponential time explicit,'' as the description of a single monomial consists of $O(k \log d)$ bits.

\begin{corollary} \label{cor:KI tradeoffs}
	Let $\F$ be a field of characteristic $p > 0$.
	Let $c, k \in \naturals$ be fixed constants.
	Let $\set{g_d(\vec{x}) : d \in \naturals}$ be a strongly $d^{O(k)}$-explicit family of $k$-variate degree $d$ polynomials which cannot be computed by circuits of size smaller than $s(d)$ over $\F^{p^{-\infty}}$.
	Then the following results hold regarding hitting sets for $\mathcal{C}_{\F}(n^c,n,n^c)$.
	\begin{enumerate}
		\item
			If $s(d) = \log^{\omega(1)}d$, then there is a $2^{n^{o(1)}}$-explicit hitting set for $\mathcal{C}_{\F}(n^c,n,n^c)$ of size $2^{n^{o(1)}}$.
		\item
			If $s(d) = 2^{\log^{\Omega(1)} d}$, then there is a $2^{\log^{O(1)} n}$-explicit hitting set for $\mathcal{C}_{\F}(n^c,n,n^c)$ of size $2^{\log^{O(1)} n}$.
		\item
			If $s(d) = d^{\Omega(1)}$, then there is a $n^{O(\log n)}$-explicit hitting set for $\mathcal{C}_{\F}(n^c,n,n^c)$ of size $n^{O(\log n)}$.
	\end{enumerate}
\end{corollary}

\begin{proof}
	Each statement follows by setting parameters in \autoref{thm:KI hard vs rand} and \autoref{cor:KI hitting set} and using the fact that $c$ and $k$ are fixed constants independent of $n$ and $d$.
	We omit the straightforward calculations.
\end{proof}

\section{Bootstrapping from Constant-Variate Hardness} \label{sec:bootstrap}

Given that we use the seemingly stronger assumption of constant-variate hardness in our extension of the Kabanets-Impagliazzo generator, one may wonder if we can push the hardness-randomness connection further and obtain a better derandomization of identity testing for $\mathcal{C}_{\F}(n^c,n,n^c)$.
Perhaps surprisingly, this is possible by going through the recent development of ``bootstrapping'' for hitting sets.

\subsection{A Non-Trivial Hitting Set from Constant-Variate Hardness}

Let $n$ be a constant and let $s$ be arbitrarily large.
Suppose we have an explicit, slightly non-trivial hitting set for $\mathcal{C}_{\F}(s,n,s)$.
Then we can ``bootstrap'' the advantage this hitting set has over the trivial one in order to obtain an explicit hitting set of very small size for $\mathcal{C}_{\F}(s,s,s)$.
That is, in order to almost completely derandomize polynomial identity testing for the class of polynomials of polynomial degree computed by polynomial-size circuits, it suffices to find a non-trivial derandomization of polynomial identity testing for circuits on a constant number of variables but of arbitrary size and degree.

We remark that, throughout this section, one should read $\mathcal{C}_{\F}(s,s,s)$ as a stand-in for $\mathcal{C}_{\F}(n^c,n,n^c)$, where $c$ is a fixed constant.
This follows by taking $s = n^c$ and noting that $\mathcal{C}_{\F}(n^c,n,n^c) \subseteq \mathcal{C}_{\F}(n^c,n^c,n^c) = \mathcal{C}_{\F}(s,s,s)$.
While the following results are stated for $\mathcal{C}_{\F}(s,s,s)$, changing $s$ by at most a polynomial factor will not qualitatively affect the results we obtain.

We now formally state the bootstrapping result.
Let $\log^\star s$ denote the iterated logarithm of $s$.
That is,
\[
	\log^\star s \coloneqq
	\begin{cases}
		1 + \log^\star(\log s) & s > 1 \\
		0 & s \le 1.
	\end{cases}
\]
This version of the bootstrapping theorem is due to \textcite{KST19} and improves upon the initial work of \textcite{AGS19}.
Note that this theorem holds over all fields, including those of positive characteristic.

\begin{theorem}[\cite{KST19}] \label{thm:bootstrapping}
	Let $\F$ be any field and let $\eps > 0$ and $n \ge 2$ be constants.
	Suppose that for all sufficiently large $s$, there is an $s^{O(n)}$-explicit hitting set of size $s^{n-\eps}$ for $\mathcal{C}_{\F}(s,n,s)$.
	Then there is an $s^{\exp \circ \exp (O(\log^\star s))}$-explicit hitting set of size $s^{\exp \circ \exp(O(\log^\star s))}$ for $\mathcal{C}_{\F}(s,s,s)$.
\end{theorem}

In this section, we will use \autoref{thm:bootstrapping} to obtain a stronger derandomization of polynomial identity testing over fields of characteristic $p > 0$ under appropriate hardness assumptions.
Suppose $\set{g_d(\vec{x}) : d \in \naturals}$ is a family of strongly $d^{O(k)}$-explicit $k$-variate degree $d$ polynomials which require algebraic circuits of size $d^{\Omega(k)}$.
Using \autoref{cor:KI tradeoffs}, we can obtain a $s^{O(\log s)}$-explicit hitting set for $\mathcal{C}_{\F}(s,s,s)$ of size $s^{O(\log s)}$.
By a more careful instantiation of the Kabanets-Impagliazzo generator, we can use the hardness assumption on $g_d$ to design an explicit hitting set which satisfies the hypotheses of \autoref{thm:bootstrapping}.
This yields an explicit hitting set for $\mathcal{C}_{\F}(s,s,s)$ of size $s^{\exp \circ \exp(O(\log^\star s))}$, which greatly improves upon the size $s^{O(\log s)}$ hitting set of \autoref{cor:KI tradeoffs}.

Our argument also works for fields of characteristic zero, giving us a general theorem which converts near-optimal constant-variate hardness into near-optimal derandomization of polynomial identity testing for $\mathcal{C}_{\F}(s,s,s)$.

First, we need a technical lemma regarding lower bounds against constant-variate polynomials.
Roughly, we will show that $d^{\delta}$ lower bounds against degree $d$ constant-variate polynomials can be magnified to $d^c$ lower bounds against constant-variate polynomials for arbitrary $\delta, c > 0$.

\begin{lemma} \label{lem:hardness exponent}
	Let $\F$ be any field.
	Let $k \in \naturals$ and $c, \delta > 0$ be fixed constants.
	Let $\set{g_d(\vec{x}) : d \in \naturals}$ be a strongly $d^{O(k)}$-explicit family of $k$-variate polynomials of degree $d$.
	Suppose that for $d$ sufficiently large, $g_d$ cannot be computed by algebraic circuits of size smaller than $d^{\delta}$ over $\F$.
	Then there is a constant $m \in \naturals$ and a family $\set{h_\Delta(\vec{y}) : \Delta \in \naturals}$ of strongly $\Delta^{O(m)}$-explicit $m$-variate degree $\Delta$ polynomials such that for $\Delta$ sufficiently large, $h_\Delta$ cannot be computed by algebraic circuits of size smaller than $\Delta^{c}$ over $\F$.
\end{lemma}

\begin{proof}
	We follow the approach of \autoref{lem:kronecker}, but in base $d^{\delta/2c}+1$ as opposed to base 2.

	Without loss of generality, assume that $\delta \le 1 \le c$.
	Let $m \coloneqq \frac{2 c k}{\delta}$ and let $\vec{y} = (y_{1,1},\ldots,y_{k,2c/\delta})$.
	Let $\sigma(y_{i,j}) = x_i^{(d^{\delta / 2c}+1)^j}$.
	We will take $h_\Delta(\vec{y})$ to be the polynomial of individual degree $d^{\delta/2c}$ which satisfies the equation $h(\sigma(\vec{y})) = g_d(\vec{x})$.
	More explicitly, let $g_d(\vec{x}) = \sum_{\vec{a} \in \naturals^k} \alpha_{\vec{a}} \vec{x}^{\vec{a}}$ be the expression of $g_d$ as a sum of monomials.
	Let $\varphi : \llb d^{\delta/2c}+1 \rrb^{2c/\delta} \to \llb d+1 \rrb$ be the map which takes the base-$(d^{\delta/2c}+1)$ expansion of a number $t \in \llb d+1 \rrb$ and returns $t$.
	Then we define $h_\Delta(\vec{y})$ as
	\[
		h_\Delta(\vec{y}) = \sum_{A \in \llb d^{\delta/2c}+1 \rrb^{k \times 2c/\delta}} \alpha_{\varphi(A_{1,\bullet}),\ldots,\varphi(A_{k,\bullet})} \prod_{i,j \in \llb d^{\delta/2c}+1 \rrb} y_{i,j}^{A_{i,j}}.
	\]
	It is clear from the construction of $h_\Delta$ that $h_\Delta(\sigma(\vec{y})) = g_d(\vec{x})$.
	The polynomial $h_\Delta$ is of individual degree at most $d^{\delta/2c}$, so $\Delta \coloneqq \deg(h_\Delta)$ can be bounded as 
	\[
		\Delta \le m d^{\delta / 2c} = \frac{2c k d^{\delta/2c}}{\delta}.
	\]
	Since $k$ and $\delta$ are fixed constants, for $d$ large enough, we obtain $\Delta \le d^{2\delta/3c}$.

	To show that $h_\Delta$ has the claimed hardness, suppose we are given a circuit of size $s$ which computes $h_\Delta$.
	By repeated squaring, we may compute the map $\sigma(\vec{y})$ using a circuit of size $O(k \log d) = O(m \log \Delta) = O(\log \Delta)$.
	This yields a circuit of size $s' \le s + O(\log \Delta)$ which computes $g_d$.
	By the assumed hardness of $g_d$, we have $s' \ge d^{\delta}$.
	Putting things together gives us
	\[
		s \ge d^{\delta} - O(\log \Delta).
	\]
	Since $\Delta \le d^{2\delta/3c}$ for $d$ large enough, we obtain
	\[
		s \ge \Delta^{3c/2} - O(\log \Delta).
	\]
	For $\Delta$ (and hence $d$) large enough, we have $s \ge \Delta^{c}$, which yields the desired lower bound on $h_\Delta$.

	It remains to verify the explicitness of $h_\Delta$.
	We can compute a coefficient of $h_\Delta$ by computing the corresponding coefficient of $g_d$, so $h_\Delta$ inherits the strong $d^{O(k)}$-explicitness of $g_d$.
	We need to show that $d^{O(k)} \le \Delta^{O(m)}$ in order to conclude that $h_\Delta$ is strongly $\Delta^{O(m)}$-explicit.
	By writing $h_\Delta$ as a sum of monomials, there is a circuit of size $\Delta^{O(m)}$ which computes $h_\Delta$.
	Combined with the argument above, this yields a circuit of size $\Delta^{O(m)} + O(\log \Delta) = \Delta^{O(m)}$ which computes $g_d$.
	Since any circuit which computes $g_d$ must have size $d^\delta$, we obtain $\Delta^{O(m)} \ge d^\delta$.
	As $c$, $k$, $\delta$, and $m$ are all fixed constants, this yields $d^{O(k)} \le \Delta^{O(m)}$ as desired.
\end{proof}

Now we are ready to state and prove our hardness-randomness tradeoff.

\begin{theorem} \label{thm:bootstrap generator}
	Let $\F$ be any field and let $k \in \naturals$ and $\delta > 0$ be fixed constants.
	Let $\mathbb{K} = \F^{p^{-\infty}}$ if $\ch \F = p > 0$ and $\mathbb{K} = \F$ otherwise.
	Let $\set{g_d(\vec{x}) \in \F[\vec{x}] : d \in \naturals}$ be a family of strongly $d^{O(k)}$-explicit $k$-variate degree $d$ polynomials.
	Suppose that for all $d$ sufficiently large, $g_d$ cannot be computed by algebraic circuits of size smaller than $d^{\delta}$ over $\mathbb{K}$.
	Then for all sufficiently large $s$, there is an $s^{\exp \circ \exp(O(\log^\star s))}$-explicit hitting set of size $s^{\exp \circ \exp(O(\log^\star s))}$ for $\mathcal{C}_{\F}(s,s,s)$.
\end{theorem}

\begin{proof}
	Using \autoref{lem:hardness exponent}, we may assume without loss of generality that $\delta \ge 30$.

	By \autoref{thm:bootstrapping}, it suffices to provide an explicit hitting set of size $s^{n - \eps}$ for $\mathcal{C}_{\F}(s,n,s)$ for constants $\eps, n$ and all $s$ sufficiently large.
	We will instantiate the Kabanets-Impagliazzo generator with $g_d$ as the hard polynomial, using the finer-grained designs of \autoref{lem:RS design}.

	Let $s$ be given.
	By adding auxiliary variables if necessary, we may assume that $k$ is a prime power.
	Note there is always a power of $2$ between $k$ and $2k$, so this at most doubles the number of variables in $g_d$.
	We set parameters as follows:
	\begin{itemize}
		\item
			$c \coloneqq 3$,
		\item
			$n \coloneqq 2k^{c+1} = 2 k^4$,
		\item
			$r \coloneqq 2$, and
		\item
			$d \coloneqq s^k$.
	\end{itemize}
	By \autoref{lem:RS design}, we can construct in $\poly(n)$ time a collection of sets $S_1,\ldots,S_n \subseteq [k^c]$ such that $|S_i| = k$ and $|S_i \cap S_j| \le r$.

	Consider the generator $\mathcal{G} : \F^{k^c} \to \F^n$ given by
	\[
		\mathcal{G}(\vec{z}) = (g_d(\vec{z}|_{S_1}),\ldots,g_d(\vec{z}|_{S_n})).
	\]
	By construction, $\mathcal{G}$ has seed length $k^c$ and degree $d = s^k$.
	Since $g_d$ is strongly $d^{O(k)}$-explicit, we can evaluate $\mathcal{G}$ by constructing the design $S_1,\ldots,S_n$, computing the coefficients of $g_d$, and evaluating each of the $n$ copies of $g_d$.
	Constructing the design takes $n^{O(1)}$ time and computing the coefficients of $g_d$ takes $d^{O(k)}$ time.
	To evaluate $g_d$, we use the expression of $g_d$ as a sum of monomials, which requires $d^{O(k)}$ time for each of the $n$ evaluations.
	In total, we can evaluate $\mathcal{G}$ in time
	\[
		n^{O(1)} \cdot d^{O(k)} = n^{O(1)} \cdot s^{O(k^2)} = n^{O(1)} \cdot s^{O(\sqrt{n})},
	\]
	so $\mathcal{G}$ is $s^{O(\sqrt{n})}$-explicit for $s$ sufficiently large.

	If $\mathcal{G}$ is in fact a hitting set generator for $\mathcal{C}_{\F}(s,n,s)$, then using \autoref{lem:hsg to hitting set}, we obtain a hitting set $\mathcal{H}$ for $\mathcal{C}_{\F}(s,n,s)$ of size
	\[
		(s \cdot d)^{k^c} = (s^{k+1})^{k^3} = s^{k^4 + k^3} \le s^{2 k^{4} - \eps} = s^{n - \eps}
	\]
	for some $\eps > 0$ when $s$ is large enough.
	Moreover, $\mathcal{H}$ is $s^{O(\sqrt{n})} \cdot |\mathcal{H}| \le s^{O(n)}$-explicit.
	We now apply \autoref{thm:bootstrapping} to obtain the claimed $s^{\exp \circ \exp(O(\log^\star s))}$-explicit hitting set for $\mathcal{C}_{\F}(s,s,s)$ of size $s^{\exp \circ \exp(O(\log^\star s))}$.
	It remains to show that $\mathcal{G}$ is indeed a hitting set generator for $\mathcal{C}_{\F}(s,n,s)$.

	To show this, suppose for the sake of contradiction that $\mathcal{G}$ is not a hitting set generator for $\mathcal{C}_{\F}(s,n,s)$. 
	Then there is some $f(\vec{y}) \in \mathcal{C}_{\F}(s,n,s)$ such that $f(\vec{y}) \neq 0$ and $f(\mathcal{G}(\vec{z})) = 0$.
	We define the hybrid polynomials $f_0,\ldots,f_n$ by
	\begin{align*}
		f_0(\vec{y},\vec{z}) &= f(y_1,\ldots,y_n) \\
		f_1(\vec{y},\vec{z}) &= f(g_d(\vec{z}|_{S_1}), y_2,\ldots, y_n) \\
		&\vdots \\
		f_{n-1}(\vec{y},\vec{z}) &= f(g_d(\vec{z}|_{S_1}),\ldots,g_d(\vec{z}|_{S_{n-1}}),y_n) \\
		f_n(\vec{y},\vec{z}) &= f(g_d(\vec{z}|_{S_1}),\ldots,g_d(\vec{z}|_{S_n})) = f(\mathcal{G}(\vec{z})).
	\end{align*}
	Since $f_0 \neq 0$ and $f_n = 0$, there is some $i \in [n]$ such that $f_{i-1} \neq 0$ and $f_i = 0$.
	Assuming $|\F| > sd \ge \deg(f_i)$, we can find an assignment to the variables $\set{y_j : j \neq i}$ and $\set{z_j : j \notin S_i}$ such that $f_i$ remains non-zero under this partial evaluation.
	If $\F$ is too small, we may find such an assignment using values from some finite extension $\F' \supseteq \F$ of size at least $sd + 1$ (and hence degree $O(\log(sd))$).
	After renaming variables, denote this non-zero restriction of $f_i$ by $\overline{f}(z_1,\ldots,z_k,y)$.

	We can compute $\overline{f}$ by composing the circuit for $f$ with at most $n-1$ copies of the partial evaluation of $g_d(\vec{z}|_{S_j})$ for $j < i$.
	By assumption, we can compute $f$ with a circuit of size $s$.
	Since $|S_j \cap S_i| \le 2$ for $j \neq i$, at most $2$ variables in $\vec{z}|_{S_j}$ are unset.
	This implies each restriction of $g_d(\vec{z}|_{S_j})$ is a polynomial of degree $d$ on 2 variables and thus can be computed by a depth-two circuit of size at most $d \cdot (d+1)^2$.
	This yields a circuit for $\overline{f}$ of size at most $s + n d\cdot (d+1)^2$.
	Note that the degree of $\overline{f}$ is bounded by $sd$, since $\overline{f}$ is the composition of two polynomials of degrees at most $s$ and $d$.

	By assumption, we have that $\overline{f}(z_1,\ldots,z_k,y) \neq 0$ and $\overline{f}(z_1,\ldots,z_k,g_d(\vec{z})) = 0$.
	This implies that $y - g_d(\vec{z})$ is a factor of $\overline{f}$.
	We now apply \autoref{thm:Kaltofen factoring} to factor the circuit for $\overline{f}$.
	\begin{itemize}
		\item
			If $\ch \F = p > 0$, we obtain a circuit for $(y - g_d(\vec{z}))^{p^t} = y^{p^t} - g_d(\vec{z})^{p^t}$ for some $t \in \naturals$.
			Since $y^{p^t} - g_d(\vec{z})^{p^t}$ is a factor of $\overline{f}(z_1,\ldots,z_k,y)$, we must have 
			\[
				d p^t = \deg(y^{p^t} - g_d(\vec{z})^{p^t}) \le \deg(\overline{f}) \le sd.
			\]
			This implies $p^t \le s$.
			Since $\overline{f}$ has degree $sd$ and is computable in size $s + O(n d^3)$, the circuit computing $y^{p^t} - g_d(\vec{z})^{p^t}$ has size at most $O((nsd)^{12})$.
			By setting $y = 0$ and negating the output of the circuit, we obtain a circuit for $g_d(\vec{z})^{p^t}$ of size $O((nsd)^{12})$.

			We now apply \autoref{cor:pth root of circuit} a total of $t$ times. 
			This produces a circuit which computes $g_d(\vec{z})$ and has size $O((nsd)^{12} p^{2kt} 2^{kt} 3^t) = O((nsd)^{12} s^{3k + 2})$.
			Here we use the fact that $p \ge 2$, so $2^{kt} \le p^{kt} \le s^k$ and $3^t \le 4^t \le p^{2t} \le s^2$.

			In the case where $|\F| > sd$, the circuit for $\overline{f}$ was defined over $\F$, so the circuit for $g_d$ is defined over $\mathbb{K} = \F^{p^{-\infty}}$.
			If instead $|\F| \le sd$, the circuit for $\overline{f}$ was defined over a finite extension $\F' \supseteq \F$ of degree $O(\log(sd))$.
			As $\F'$ is a finite field, $\F'$ is perfect, so the circuit obtained from \autoref{cor:pth root of circuit} is defined over $\F'$.
			We apply \autoref{lem:simulate extension} to simulate this circuit over $\F$, incurring an extra $O(\log^3(sd))$ factor in the circuit size.

			In total, we now have a circuit which computes $g_d$ over $\mathbb{K} = \F^{p^{-\infty}}$ and has size bounded by $O((nsd)^{12} s^{3k + 2} \log^3(sd))$.

		\item
			If $\ch \F = 0$, the previous case applies, but without the need to take a $p$\ts{th} root or simulate a field extension.
			This yields a circuit which computes $g_d(\vec{z})$ over $\mathbb{K} = \mathbb{F}$ and has size $O((nsd)^{12})$.

	\end{itemize}
	In both cases, we obtain a circuit which computes $g_d(\vec{z})$ over $\mathbb{K}$ and has size at most $O((nsd)^{12} s^{3k + 2}\log^3(sd))$.
	Restating in terms of $k$ and $d$, we have a circuit for $g_d$ of size
	\[
		O((nsd)^{12} s^{3k + 2}\log^3(sd)) = O(k^{48} s^{14+3k}d^{12} \log^3(d)) = O(k^{48} d^{15 + 14/k} \log^3(d)).
	\]
	Since $k \ge 1$ and $k$ is a constant, we can bound the size of the circuit computing $g_d$ by $O(d^{29} \log^3(d))$.
	This contradicts the fact that $g_d$ requires circuits over $\mathbb{K}$ of size $d^{\delta} \ge d^{30}$ for sufficiently large $d$.
	Hence $\mathcal{G}$ is in fact a hitting set generator for $\mathcal{C}_{\F}(s,n,s)$.
\end{proof}

\subsection{Comparison to Characteristic Zero} \label{subsec:compare to GKSS}

Over fields of characteristic zero, the recent work of \textcite{GKSS19} obtained what is currently the best-known derandomization of polynomial identity testing for $\mathcal{C}_{\F}(s,s,s)$ under a hardness assumption.
From an explicit family of $k$-variate degree $d$ polynomials of hardness $d^{\Omega(1)}$, they obtain an explicit hitting set for $\mathcal{C}_{\F}(s,s,s)$ of size $s^{O(1)}$.
Specifically, they prove the following theorem.

\begin{theorem}[\cite{GKSS19}] \label{thm:GKSS generator}
	Let $\F$ be a field of characteristic zero.
	Let $k \in \naturals$ be large enough and let $\delta > 0$ be a fixed constant.
	Suppose $\set{P_{k,d} \in \F[\vec{x}] : d \in \naturals}$ is a family of $d^{O(k)}$-explicit $k$-variate polynomials of degree $d$ such that $P_{k,d}$ cannot be computed by algebraic circuits of size smaller than $d^{\delta}$.
	Then there is an $s^{(k/\delta)^{O(1)}}$-explicit hitting set for $\mathcal{C}_{\F}(s,s,s)$ of size $s^{O(k^2/\delta^2)}$.
\end{theorem}

We remark that \textcite{GKSS19} do not define the notion of explicitness they use in their result, but it is enough for $P_{k,d}$ to be computable by a uniform algorithm which runs in time $d^{O(k)}$.
This is slightly different from our notion of strong explicitness, where we require the coefficients of $P_{k,d}$ to be computable in $d^{O(k)}$ time.
It is clear that one can pass from strong explicitness to the standard notion of explicitness by computing a polynomial as a sum of monomials.
Via polynomial interpolation, one can show that polynomials which are ``evaluation-explicit'' are strongly explicit.
In both cases, the explicitness parameter may degrade considerably, as the number of terms in a polynomial may be exponentially larger than the amount of time required to compute the polynomial or one of its coefficients.
In general, one cannot hope to do better than this: in one direction, the coefficients of the permanent are easy to compute, but the permanent is widely conjectured to be hard to compute; in the other direction, there are examples of polynomials which are easy to compute but which have the permanent of a large matrix embedded in their coefficients (see, for example, \textcite[\textsection 2.3]{Bur00}).

In the context of \autoref{thm:bootstrap generator} and \autoref{thm:GKSS generator}, however, the two notions of explicitness coincide.
When working with $k$-variate polynomials of degree $d$, we incur an overhead of $d^{O(k)}$ in passing between the two notions of explicitness.
As the hypotheses of these theorems are already in the regime of (strong) $d^{O(k)}$-explicitness, the explicitness parameter changes by a polynomial factor, which is small enough to not affect the asymptotics of the results obtained.

The fact that the underlying field has characteristic zero is used in a key part of the proof of \autoref{thm:GKSS generator}, and it is not clear how to adapt the proof to fields of positive characteristic.
The generator used to design the hitting set in the conclusion of \autoref{thm:GKSS generator} is notably not a variation on the Kabanets-Impagliazzo generator, but instead a new generator whose construction is more algebraic than combinatorial in flavor.

Note that \autoref{thm:bootstrap generator} and \autoref{thm:GKSS generator} require the same hardness assumption.
This gives a second proof of derandomization of polynomial identity testing from an explicit family of hard constant-variate polynomials, although the derandomization we obtain is slightly weaker compared to \autoref{thm:GKSS generator}.
However, our construction does not require the characteristic of the underlying field to be zero.
It is tempting to conjecture that one can recover the conclusion of \autoref{thm:GKSS generator} in positive characteristic by improving the bootstrapping process used to prove \autoref{thm:bootstrapping}.
It is unclear whether such a result is possible.

\section{Relating Constant-Variate and Multivariate Lower Bounds} \label{sec:kronecker}

This work and the work of \textcite{GKSS19} have shown that lower bounds against (strongly) explicit constant-variate polynomials yield very strong derandomizations of polynomial identity testing.
We are able to give an explicit hitting set of size $s^{\exp \circ \exp(O(\log^\star s))}$ for $\mathcal{C}_{\F}(s,s,s)$ for any field $\F$ (this is \autoref{thm:bootstrap generator}), while \textcite{GKSS19} obtain explicit hitting sets of size $s^{O(1)}$ for the same class when $\ch\F = 0$.
However, if one instead assumes the existence of a (strongly) explicit family of maximally-hard multivariate polynomials of low degree (specifically, degree $n^{O(1)}$ where $n$ is the number of variables), it is not clear how to obtain similar derandomization results.
The best-known derandomization from multivariate lower bounds is that of \textcite{KI04}, who gave an explicit hitting set of size $s^{O(\log s)}$ for $\mathcal{C}_{\F}(s,s,s)$.

The fact that we can obtain strong derandomizations of polynomial identity testing from constant-variate hardness raises the question of whether or not such derandomization is possible under multivariate hardness assumptions.
A natural first approach to this would be to show that lower bounds for a (strongly) explicit family of multivariate polynomials imply comparable lower bounds against a (strongly) explicit family of constant-variate polynomials.
Such an implication is known in the setting of non-commutative circuits and is due to \textcite{CILM18}.

It is not hard to show a connection in the other direction; that is, lower bounds against strongly explicit families of constant-variate polynomials can be translated into comparable lower bounds against strongly explicit families of multivariate polynomials.
An easy way to do this is via the approach of \autoref{lem:kronecker}.

In this section, we investigate to what extent a converse to \autoref{lem:kronecker} may hold.
Unconditionally refuting the converse of \autoref{lem:kronecker} requires proving circuit lower bounds that seem far out of reach, so we have little hope to fully resolve this question.
However, we can give some complexity-theoretic evidence which shows a converse to \autoref{lem:kronecker} is unlikely to hold.
To do this, we take a detour into the arithmetic complexity of integers.

\subsection{Complexity of Computing Integers}

We start by defining the model we use to compute sequences of integers.

\begin{definition}
	For a natural number $n \in \naturals$, let $\tau(n)$ denote the size of the smallest circuit which computes $n$ using the constant $1$ and the operations of addition, subtraction, and multiplication.
	Let $(a_n)_{n \in \naturals}$ be a sequence of natural numbers.
	If $\tau(a_n) \le \log^{O(1)}n$, then we say $(a_n)_{n \in \naturals}$ is \emph{easy to compute}.
	Otherwise, we say $(a_n)_{n \in \naturals}$ is \emph{hard to compute}.
\end{definition}

As an example, the sequence $(2^n)_{n \in \naturals}$ is easy to compute, as we can compute $2^n$ in $O(\log n)$ arithmetic steps by repeated squaring.
A major open problem in this area is to understand $\tau(n!)$, the complexity of the sequence of factorials.
The following conjecture regarding $\tau(n!)$ appears to be folklore. 

\begin{conjecture} \label{conj:factorials are hard}
	The sequence of factorials $(n!)_{n \in \naturals}$ is hard to compute.
\end{conjecture}

Prior work has established relationships between \autoref{conj:factorials are hard} and other prominent conjectures in computational complexity.
\textcite[page 126]{BCSS98} gave an argument that shows if $\tau(n!) \le \log^{O(1)} n$, then there are circuits of $\log^{O(1)} n$ size to factor $n$.
A related work by \textcite{Sha79} reduces factorization to computing factorials, albeit in a slightly different model.
\textcite{Bur09} showed that \autoref{conj:factorials are hard} implies that the $n \times n$ permanent cannot be computed by constant-free division-free algebraic circuits of size $n^{O(1)}$.
Work by \textcite{Lip94} shows that average-case hardness of factoring implies a slightly weaker form of \autoref{conj:factorials are hard}; namely, that the polynomial $\prod_{i=1}^n (x - i)$ is hard to compute by constant-free algebraic circuits.

Before moving on to address the question of a converse to \autoref{lem:kronecker}, we present a reduction due to \textcite{Sha79} which reduces the task of computing $n!$ to the task of computing $\binom{2n}{n}$.

\begin{lemma}[\cite{Sha79}] \label{lem:Shamir reduction}
	If $(\binom{2n}{n})_{n \in \naturals}$ is easy to compute, then $(n!)_{n \in \naturals}$ is easy to compute.
\end{lemma}

\begin{proof}
	Suppose $\tau\del{\binom{2n}{n}} \le O(\log^c n)$.
	Recall the identity
	\[
		n! = 
		\begin{cases}
			((n/2)!)^2 \cdot \binom{n}{n/2}& \text{$n$ is even} \\
			n \cdot ((\frac{n-1}{2})!)^2 \cdot \binom{n-1}{(n-1)/2} & \text{$n$ is odd}. 
		\end{cases}
	\]
	This implies
	\[
		\tau(n!) \le \tau(n) + \tau((\floor{n/2}!)^2) + \tau\del{\binom{2 \cdot \floor{n/2}}{\floor{n/2}}}.
	\]
	Expanding out the recurrence and using the fact that $\tau((\floor{n/2}!)^2) \le \tau(\floor{n/2}!) + 1$, we get
	\begin{align*}
		\tau(n!) &\le \sum_{i=1}^{\log n} \sbr{\tau(\floor{n/2^i}) + \tau\del{\binom{2 \cdot \floor{n/2^{i+1}}}{\floor{n/2^{i+1}}}} + 1} \\
		& \le \log n \cdot \del{O(\log n) + O(\log^c n) + 1} \\
		& \le O(\log^{c+1} n).
	\end{align*}
	Hence $(n!)_{n \in \naturals}$ is easy to compute.
\end{proof}

\subsection{The Inverse Kronecker Map and Constant-Free Circuits}

Here, we show that two forms of a converse to \autoref{lem:kronecker} refute \autoref{conj:factorials are hard} to varying degrees.
Our first argument shows that a straightforward converse of \autoref{lem:kronecker} implies that \autoref{conj:factorials are hard} fails infinitely often.
That is, suppose $g(x)$ is a univariate degree $d$ polynomial and $f(\vec{y})$ is a multilinear polynomial which simplifies to $g(x)$ under the mapping $y_i \mapsto x^{2^i}$.
\autoref{lem:kronecker} says that hardness of $g(x)$ implies hardness of $f(\vec{y})$.
The following conjecture, which we wish to conditionally refute, says that hardness of $f(\vec{y})$ implies hardness of $g(x)$.

\begin{conjecture} \label{conj:inverse kronecker}
	Let $g_{m,d}(\vec{x}) = \sum_{\vec{a}} \alpha_{\vec{a}} \vec{x}^{\vec{a}}$ be an $m$-variate degree $d$ polynomial.
	Let $j : \bits^{\floor{\log d}+1} \to \llb 2^{\floor{\log d}+1} \rrb$ be given by $j(\vec{e}) = \sum_{i=1}^{\floor{\log d}+1} \vec{e}_i 2^{i-1}$.
	That is, $j(\vec{e})$ is the number whose binary representation corresponds to $\vec{e}$.
	Let $\vec{y} = (y_{1,1},\ldots,y_{1,\floor{\log d}+1},\ \ldots\ ,y_{m,1},\ldots,y_{m,\floor{\log d}+1})$ and define
	\[
		f_{m,d}(\vec{y}) = \sum_{\vec{e} \in \bits^{m \times \floor{\log d}+1}} \alpha_{(j(\vec{e}_{1,\bullet}),\ldots,j(\vec{e}_{m,\bullet}))} \vec{y}^{\vec{e}}.
	\]
	Suppose $f_{m,d}$ requires constant-free circuits of size $s$ to compute.
	Then $g_{m,d}$ requires constant-free circuits of size $s^{\Omega(1)} - \Theta(m \log d)$ to compute.
\end{conjecture}

We now show that \autoref{conj:inverse kronecker} implies the factorials are easy to compute infinitely often.

\begin{theorem} \label{thm:factorials are easy i-o}
	Suppose \autoref{conj:inverse kronecker} holds over $\rationals$.
	Then the sequence of factorials $(n!)_{n \in \naturals}$ is easy to compute infinitely often.
\end{theorem}

\begin{proof}
	It is easy to see that $\sum_{i=0}^{2^n} \binom{2^n}{i} x^i = (x+1)^{2^n}$ is computable by a constant-free algebraic circuit of size $O(n)$ via repeated squaring.
	Let
	\[
		f_n(\vec{y}) = \sum_{\vec{e} \in \bits^{n+1}} \binom{2^n}{j(\vec{e})} \vec{y}^{\vec{e}}.
	\]
	The contrapositive of \autoref{conj:inverse kronecker} yields a constant-free circuit of size $O(n^c)$ which computes $f_n$ for some absolute constant $c$.
	Let $a_{n-1} = 1$ and $a_0 = \cdots = a_{n-2} = a_n = 0$.
	Then $f_n(\vec{a}) = \binom{2^n}{2^{n-1}} + 1$.
	By evaluating the circuit for $f_n$ at $\vec{a}$ and subtracting $1$, we obtain a circuit of size $O(n^c)$ which computes $\binom{2^n}{2^{n-1}}$.
	
	We now follow the argument of \autoref{lem:Shamir reduction} to construct circuits of size $O(n^{c+1})$ to compute $(2^n!)_{n \in \naturals}$.
	By definition, we have
	\begin{align*}
		2^n! &= \binom{2^n}{2^{n-1}} (2^{n-1}!)^2 \\
		&= \binom{2^n}{2^{n-1}} \binom{2^{n-1}}{2^{n-2}}^2 (2^{n-2}!)^4 \\
		&\vdots \\
		&= \prod_{i=0}^{n-1} \binom{2^{n-i}}{2^{n-i-1}}^{2^i}.
	\end{align*}
	Using the fact that we fact that we can compute $\binom{2^n}{2^{n-1}}$ by a circuit of size $O(n^c)$, we obtain
	\[
		\tau(2^n!) \le \sum_{i=0}^{n-1} \tau\del{\binom{2^{n-i}}{2^{n-i-1}}^{2^i}} \le \sum_{i=0}^{n-1} O(n^{c+1}) \le O(n^{c+2}).
	\]
	Hence the factorials are easy to compute infinitely often.
\end{proof}

It is unclear whether there is meaningful evidence to suggest that the factorials are not easy to compute at numbers of the form $2^n$.
Because of this, \autoref{thm:factorials are easy i-o} may be best viewed as evidence that if \autoref{conj:inverse kronecker} is true, the proof will not be straightforward.

\autoref{conj:inverse kronecker} can be seen as a base-two converse to \autoref{lem:kronecker}.
Instead, we might consider the following strengthening of \autoref{conj:inverse kronecker} to all number bases.

\begin{conjecture} \label{conj:strong inverse kronecker}
	Let $g_{m,d}(\vec{x}) = \sum_{\vec{a}} \alpha_{\vec{a}} \vec{x}^{\vec{a}}$ be an $m$-variate degree $d$ polynomial.
	Let $k \in \naturals$ and let $j : \llb k \rrb^{\floor{\log_k d}+1} \to \llb k^{\floor{\log_k d}+1} \rrb$ be given by $j(\vec{e}) = \sum_{i=1}^{\floor{\log_k d}+1} \vec{e}_i k^{i-1}$, that is, $j(\vec{e})$ is the number whose base-$k$ representation corresponds to $\vec{e}$.
	Let $\vec{y} = (y_{1,1},\ldots,y_{1,\floor{\log_k d}+1},\ \ldots\ ,y_{m,1},\ldots,y_{m,\floor{\log_k d}+1})$ and define
	\[
		f_{m,d}(\vec{y}) = \sum_{\vec{e} \in \llb k \rrb^{m \times \floor{\log_k d}+1}} \alpha_{(j(\vec{e}_{1,\bullet}),\ldots,j(\vec{e}_{m,\bullet}))} \vec{y}^{\vec{e}}.
	\]
	Suppose $f_{m,d}$ requires constant-free circuits of size $s$ to compute.
	Then $g_{m,d}$ requires constant-free circuits of size $s^{\Omega(1)} - \Theta(m \log d)$ to compute.
\end{conjecture}

We can show that this stronger conjecture is less likely to hold than \autoref{conj:inverse kronecker}.

\begin{theorem}
	Suppose \autoref{conj:strong inverse kronecker} holds over $\rationals$.
	Then $(n!)_{n \in \naturals}$ is easy to compute.
\end{theorem}

\begin{proof}
	By \autoref{lem:Shamir reduction}, it suffices to show that the central binomial coefficients $\binom{2n}{n}_{n \in \naturals}$ are easy to compute.
	Let $n \in \naturals$ be given.
	There is constant-free circuit of size $O(\log n)$ which computes $g(x) = (x + 1)^{2n}$.
	Consider the polynomial
	\[
		f(y_1,y_n) = \sum_{i = 0}^{n-1} \sum_{j=0}^{n-1} \binom{2n}{i + jn} y_1^i y_n^j,
	\]
	where by convention $\binom{n}{k} = 0$ when $n < k$.
	Note that
	\[
		f(x,x^n) =  \sum_{i=0}^{n-1} \sum_{j=0}^{n-1} \binom{2n}{i+jn}x^{i + jn} = \sum_{k=0}^{n^2-1}\binom{2n}{k}x^k = \sum_{k=0}^{2n} \binom{2n}{k}x^k = (x+1)^{2n}.
	\]
	The contrapositive of \autoref{conj:strong inverse kronecker} implies that $f$ is computable by a constant-free circuit of size $O(\log^c n)$ for some absolute constant $c$.
	We now evaluate $f(0,1)$ to obtain
	\[
		f(0,1) = \sum_{j=0}^{n-1} \binom{2n}{jn} = \binom{2n}{0} + \binom{2n}{n} + \binom{2n}{2n} = \binom{2n}{n} + 2.
	\]
	By computing $f(0,1) - 2$, we obtain a constant-free circuit of size $O(\log^c n)$ which computes $\binom{2n}{n}$.
	Hence the central binomial coefficients are easy to compute.
\end{proof}

Note that the results of this section only give evidence that \autoref{conj:inverse kronecker} and \autoref{conj:strong inverse kronecker} do not hold over fields of characteristic zero.
Over fields of positive characteristic, it is unclear whether these conjectures are likely to be true or false.
This is somewhat interesting, as if \autoref{conj:inverse kronecker} holds over fields of positive characteristic, then we can replace constant-variate hardness with multivariate hardness in our extension of the Kabanets-Impagliazzo generator to fields of small characteristic.

\section{Conclusion and Open Problems} \label{sec:conclusion}

In this work, we gave the first instantiation of the algebraic hardness-randomness paradigm over fields of small characteristic.
Our main tool was the mod-$p$ decomposition, which we used to efficiently compute $p$\ts{th} roots of circuits which depend on a small number of variables.
This allowed us to extend known hardness-randomness tradeoffs due to \textcite{KI04} to fields of small characteristic under seemingly stronger hardness assumptions.
We also constructed a hitting set generator which, under suitable hardness assumptions, provides a near-complete derandomization of polynomial identity testing.
As our hardness assumptions are somewhat atypical, we compared them to more standard hardness assumptions and gave a conditional result which says that our hardness assumptions are not implied by standard ones.

A number of problems in low-characteristic derandomization remain open, some of which we have pointed out earlier in this work.
Here, we mention some challenges which our techniques are not able to resolve.
\begin{enumerate}
	\item
		Is it possible to obtain hardness-randomness tradeoffs over fields of small characteristic using a strongly explicit family of hard multilinear polynomials as opposed to constant-variate polynomials?
	\item
		Let $\F$ be a field of characteristic $p > 0$, where $p$ is some fixed constant.
		Suppose $f(\vec{x})^p\in\F[\vec{x}]$ is an $n$-variate polynomial which can be computed by a circuit of size $s$ over $\F$.
		Is there a circuit of size $s^{O(1)}$ which computes $f(\vec{x})$ in the case that $n = \omega(\log s)$?
	\item
		In the conclusion of \autoref{thm:bootstrapping}, is it possible to obtain a hitting set of size $s^{O(1)}$?
		If so, this would give a construction of a hitting set generator over low characteristic fields which qualitatively matches the parameters of the generator of \textcite{GKSS19}.
	\item
		Is it possible to lift lower bounds from the multivariate regime to the constant-variate regime?
		It seems like the answer may be ``no,'' but our evidence thus far only applies to constant-free circuits over fields of characteristic zero.
		What can we say if we remove the constant-free restriction?
		What about fields of positive characteristic?
\end{enumerate}

\textbf{Acknowledgements.}
We would like to thank Michael A.~Forbes for many useful comments which helped improve the presentation of this work.

\printbibliography

\end{document}